\documentclass[11pt]{article}
\usepackage{amsmath,amsthm,amssymb}
\usepackage{cite}
\usepackage[a4paper,margin=2.5cm]{geometry}
\usepackage{authblk}
\usepackage{array}
\usepackage{stackrel}
\usepackage{tikz}
\usetikzlibrary{calc,decorations.pathmorphing,decorations.markings,decorations.pathreplacing,patterns,shapes,arrows}

\usetikzlibrary{calc,decorations.pathmorphing,decorations.markings,
  decorations.pathreplacing,patterns,shapes}

\tikzset{boper/.style={rectangle,fill,inner sep=2pt,black}}
\tikzset{bblob/.style={circle,fill,inner sep=1.5pt,black}}
\tikzset{blob/.style={circle,draw,fill=white,outer sep=0mm,inner sep=0.3mm}}

\tikzset{mblob/.style={circle,fill=white,draw=black,inner sep=1.5pt}}
\tikzset{diam/.style={diamond,fill,inner sep=1.5pt,black}}
\tikzset{pblob/.style={circle,fill=white,draw=red,inner sep=1.5pt}}

\tikzset{
  aline/.style={
    decorate,
    decoration={
      meta-amplitude=#1,
      meta-segment length=0.15cm,
},
    postaction={decorate,ultra thick,decoration={markings,mark = at position #1 with {\arrow{>}}}}        
  },
  aline/.default=0.9
}

\tikzset{
  bline/.style={blue,
    decorate,
    decoration={
      meta-amplitude=#1,
      meta-segment length=0.3cm,
},
    postaction={decorate,ultra thick,decoration={markings,mark = at position #1 with {\arrow{>}}}}        
  },
  bline/.default=0.9
}

\tikzset{edge/.style={very thick,draw=blue}}%
\tikzset{contour/.style={brown,dashed,postaction={decorate,decoration={markings,mark = at position #1 with {\arrow{>}}}}}}
\tikzset{contour/.default=0.5}

\def\loos{0.35}
\def\slt{0.2}
\pgfmathsetmacro{\ae}{atan(\slt)}
\pgfmathsetmacro{\aw}{\ae+180}
\pgfmathsetmacro{\an}{90-\ae}
\pgfmathsetmacro{\as}{\an+180}
\pgfmathsetmacro{\sltb}{sqrt(1-\slt*\slt)}
\pgfmathsetmacro{\lcrot}{45-atan(\slt/\sltb)*0.5}
\tikzset{distort/.style={cm={1,0,-\slt,\sltb,(0,0)}}}
\def\goI#1(#2,#3){
\pgfextra{
\pgfmathparse{#1+180}\global\let\oldangle=\currentangle\global\let\newangle=\pgfmathresult\global\let\currentangle=#1
\pgfmathparse{\oldx+#2}\global\let\newx=\pgfmathresult\xdef\oldx{#2}
\pgfmathparse{\oldy+#3}\global\let\newy=\pgfmathresult\xdef\oldy{#3}
}
.. controls ++(\oldangle:\loos) and ++(\newangle:\loos) .. ++(\newx,\newy)
}
\def\go#1{\expandafter\goI#1}
\def\startI#1(#2,#3){
\pgfextra{\global\let\currentangle=#1\xdef\oldx{#2}\xdef\oldy{#3}}}
\def\start#1{\expandafter\startI#1}
\def\north{\an(0,0.5)}
\def\south{\as(0,-0.5)}
\def\east{\ae(0.5,0)}
\def\west{\aw(-0.5,0)}

\tikzset{bgplaq/.style={fill=lightgray!20!white}}
\tikzset{dgreen/.style={green!50!black,thick}}

\def\plaq(#1,#2){
\begin{scope}[shift={(#1,#2)}]
\draw[dotted] (-0.5,-0.5) rectangle ++(1,1); 
\end{scope}
}

\def\plaqz(#1,#2){
\begin{scope}[shift={(#1,#2)}]
\draw[bgplaq,dotted] (-0.5,-0.5) rectangle ++(1,1); 
\end{scope}
}

\def\plaqa(#1,#2){
\begin{scope}[shift={(#1,#2)}]
\draw[dotted,bgplaq] (-0.5,-0.5) rectangle ++(1,1);
\draw[edge] (0,-0.5) \start\north\go\east;
\draw[edge] (0,0.5) \start\south\go\west;
\end{scope}
}
\def\plaqb(#1,#2){
\begin{scope}[shift={(#1,#2)}]
\draw[dotted,bgplaq] (-0.5,-0.5) rectangle ++(1,1);
\draw[edge] (0,0.5) \start\south\go\east;
\draw[edge] (0,-0.5) \start\north\go\west;
\end{scope}
}

\def\bplaqe(#1,#2){
\begin{scope}[shift={(#1,#2)}]
\draw[dotted] (-0.5,0.5) -- (0.5,0.5) -- (-0.5,-0.5) -- cycle;
\end{scope}
}
\def\bplaqw(#1,#2){
\begin{scope}[shift={(#1,#2)}]
\draw[dotted] (0.5,-0.5) -- (0.5,0.5) -- (-0.5,-0.5) -- cycle;
\end{scope}
}
\def\bplaqn(#1,#2){
\begin{scope}[shift={(#1,#2)}]
\draw[dotted] (0.5,-0.5) -- (-0.5,0.5) -- (-0.5,-0.5) -- cycle;
\end{scope}
}
\def\bplaqs(#1,#2){
\begin{scope}[shift={(#1,#2)}]
\draw[dotted] (0.5,-0.5) -- (-0.5,0.5) -- (0.5,0.5) -- cycle;
\end{scope}
}
\def\bplaqez(#1,#2){
\begin{scope}[shift={(#1,#2)}]
\draw[dotted,bgplaq] (-0.5,0.5) -- (0.5,0.5) -- (-0.5,-0.5) -- cycle;
\end{scope}
}
\def\bplaqwz(#1,#2){
\begin{scope}[shift={(#1,#2)}]
\draw[dotted,bgplaq] (0.5,-0.5) -- (0.5,0.5) -- (-0.5,-0.5) -- cycle;
\end{scope}
}
\def\bplaqnz(#1,#2){
\begin{scope}[shift={(#1,#2)}]
\draw[dotted,bgplaq] (0.5,-0.5) -- (-0.5,0.5) -- (-0.5,-0.5) -- cycle;
\end{scope}
}
\def\bplaqsz(#1,#2){
\begin{scope}[shift={(#1,#2)}]
\draw[dotted,bgplaq] (0.5,-0.5) -- (-0.5,0.5) -- (0.5,0.5) -- cycle;
\end{scope}
}
\def\bplaqea(#1,#2){
\begin{scope}[shift={(#1,#2)}]
\draw[dotted,bgplaq] (-0.5,0.5) -- (0.5,0.5) -- (-0.5,-0.5) -- cycle;
\draw[edge] (0,0.5) \start\south\go\west;
\end{scope}
}
\def\bplaqeb(#1,#2){
\begin{scope}[shift={(#1,#2)}]
\draw[dotted,bgplaq] (-0.5,0.5) -- (0.5,0.5) -- (-0.5,-0.5) -- cycle;
\draw[edge] (0,0.5) \start\south\go\west
node[blob] {};
\end{scope}
}
\def\bplaqwa(#1,#2){
\begin{scope}[shift={(#1,#2)}]
\draw[dotted,bgplaq] (0.5,-0.5) -- (0.5,0.5) -- (-0.5,-0.5) -- cycle;
\draw[edge] (0,-0.5) \start\north\go\east;
\end{scope}
}
\def\bplaqwb(#1,#2){
\begin{scope}[shift={(#1,#2)}]
\draw[dotted,bgplaq] (0.5,-0.5) -- (0.5,0.5) -- (-0.5,-0.5) -- cycle;
\draw[edge] (0,-0.5) \start\north\go\east
node[blob] {};
\end{scope}
}
\def\bplaqna(#1,#2){
\begin{scope}[shift={(#1,#2)}]
\draw[dotted,bgplaq] (0.5,-0.5) -- (-0.5,0.5) -- (-0.5,-0.5) -- cycle;
\draw[edge] (0,-0.5) \start\north\go\west;
\end{scope}
}
\def\bplaqnb(#1,#2){
\begin{scope}[shift={(#1,#2)}]
\draw[dotted,bgplaq] (0.5,-0.5) -- (-0.5,0.5) -- (-0.5,-0.5) -- cycle;
\draw[edge] (0,-0.5) \start\north\go\west 
node[blob] {};
\end{scope}
}
\def\bplaqsa(#1,#2){
\begin{scope}[shift={(#1,#2)}]
\draw[dotted,bgplaq] (0.5,-0.5) -- (-0.5,0.5) -- (0.5,0.5) -- cycle;
\draw[edge] (0,0.5) \start\south\go\east;
\end{scope}
}
\def\bplaqsb(#1,#2){
\begin{scope}[shift={(#1,#2)}]
\draw[dotted,bgplaq] (0.5,-0.5) -- (-0.5,0.5) -- (0.5,0.5) -- cycle;
\draw[edge] (0,0.5) \start\south\go\east
node[blob] {}; 
\end{scope}
}

\def\plaqff(#1,#2){
\begin{scope}[shift={(#1,#2)}]
\draw[dotted,bgplaq] (-0.5,-0.5) rectangle ++(1,1);
\draw[edge] (0,-0.5) \start\north\go\east;
\end{scope}
}
\def\plaqf(#1,#2){
\begin{scope}[shift={(#1,#2)}]
\draw[dotted,bgplaq] (-0.5,-0.5) rectangle ++(1,1);
\draw[edge] (0,0.5) \start\south\go\west;
\end{scope}
}
\def\plaqd(#1,#2){
\begin{scope}[shift={(#1,#2)}]
\draw[dotted,bgplaq] (-0.5,-0.5) rectangle ++(1,1);
\draw[edge] (0,-0.5) \start\north\go\west;
\end{scope}
}
\def\plaqdd(#1,#2){
\begin{scope}[shift={(#1,#2)}]
\draw[dotted,bgplaq] (-0.5,-0.5) rectangle ++(1,1);
\draw[edge] (0,0.5) \start\south\go\east;
\end{scope}
}
\def\plaqc(#1,#2){
\begin{scope}[shift={(#1,#2)}]
\draw[dotted,bgplaq] (-0.5,-0.5) rectangle ++(1,1);
\draw[edge] (0,-0.5) \start\north\go\north;
\end{scope}
}
\def\plaqcc(#1,#2){
\begin{scope}[shift={(#1,#2)}]
\draw[dotted,bgplaq] (-0.5,-0.5) rectangle ++(1,1);
\draw[edge] (-0.5,0) \start\east\go\east;
\end{scope}
}
\def\plaqg(#1,#2){
\begin{scope}[shift={(#1,#2)}]
\draw[dotted,bgplaq] (-0.5,-0.5) rectangle ++(1,1);
\end{scope}
}

\def\cross(#1,#2){
\begin{scope}[shift={(#1,#2)}]
\draw[thick,blue] (-0.1,-0.1) -- (0.1,0.1); 
\draw[thick,blue] (-0.1,0.1) -- (0.1,-0.1); 
\end{scope}
}

\def\scross(#1,#2){
\begin{scope}[shift={(#1,#2)}]
\draw[thick,blue] (-0.07,-0.07) -- (0.07,0.07); 
\draw[thick,blue] (-0.07,0.07) -- (0.07,-0.07); 
\end{scope}
}

\def\triplaq(#1,#2,#3,#4,#5){
\begin{scope}[shift={(#1,#2)}]
\draw[dgreen] (1.6,-1) -- (0,0) -- (1.6,1);
\draw(0,0) node[bblob] {};
\draw(1.6,1) node[bblob] {};
\draw(1.6,-1) node[bblob] {};
\draw(0,0) node[left] {${#3}$};
\draw(1.6,1) node[right] {${#4}$};
\draw(1.6,-1) node[right] {${#5}$};
\draw[aline=0.9] (0.8,0.5) -- ( (0.8,-0.5);
\cross(0.8,0);
\draw[blue,wavy] (1.6,0) -- (0.8,0);
\end{scope}
}

\def\recpplaq(#1,#2,#3,#4,#5,#6){
\begin{scope}[shift={(#1,#2)}]
\draw[dgreen] (-0.5,-0.5) -- (0.5,-0.5);
\draw[dgreen] (-0.5,0.5) -- (0.5,0.5);
\draw (-0.5,-0.5) node[left] {${#6}$};
\draw (0.5,-0.5)  node[right] {${#5}$};
\draw (-0.5,0.5) node[left] {${#3}$};
\draw  (0.5,0.5) node[right] {${#4}$};
\draw (-0.5,-0.5) node[bblob] {};
\draw (0.5,-0.5)  node[bblob] {};
\draw (-0.5,0.5) node[bblob] {};
\draw  (0.5,0.5) node[bblob] {};
\draw[aline=0.9] (0,0.5) -- ( (0,-0.5);
\draw[blue,wavy=0.2] (-0.5,0) -- (0.5,0);
\end{scope}
}

\def\recmplaq(#1,#2,#3,#4,#5,#6){
\begin{scope}[shift={(#1,#2)}]
\draw[dgreen] (-0.5,-0.5) -- (0.5,-0.5);
\draw[dgreen] (-0.5,0.5) -- (0.5,0.5);
\draw (-0.5,-0.5) node[left] {${#6}$};
\draw (0.5,-0.5)  node[right] {${#5}$};
\draw (-0.5,0.5) node[left] {${#3}$};
\draw  (0.5,0.5) node[right] {${#4}$};
\draw (-0.5,-0.5) node[bblob] {};
\draw (0.5,-0.5)  node[bblob] {};
\draw (-0.5,0.5) node[bblob] {};
\draw  (0.5,0.5) node[bblob] {};
\draw[aline=0.9] (0,0.5) -- ( (0,-0.5);
\draw[blue,wavy=0.2] (0.5,0) -- (-0.5,0);
\end{scope}
}

\def\recpdashplaq(#1,#2,#3,#4,#5,#6){
\begin{scope}[shift={(#1,#2)}]
\draw[dgreen,dashed] (-0.5,-0.5) -- (0.5,-0.5);
\draw[dgreen] (-0.5,0.5) -- (0.5,0.5);
\draw (-0.5,-0.5) node[left] {${#6}$};
\draw (0.5,-0.5)  node[right] {${#5}$};
\draw (-0.5,0.5) node[left] {${#3}$};
\draw  (0.5,0.5) node[right] {${#4}$};
\draw (-0.5,-0.5) node[bblob] {};
\draw (0.5,-0.5)  node[bblob] {};
\draw (-0.5,0.5) node[bblob] {};
\draw  (0.5,0.5) node[bblob] {};
\draw[aline=0.9] (0,0.5) -- ( (0,-0.5);
\draw[blue,wavy=0.2] (-0.5,0) -- (0.5,0);
\end{scope}
}

\def\recmdashplaq(#1,#2,#3,#4,#5,#6){
\begin{scope}[shift={(#1,#2)}]
\draw[dgreen,dashed] (-0.5,-0.5) -- (0.5,-0.5);
\draw[dgreen] (-0.5,0.5) -- (0.5,0.5);
\draw (-0.5,-0.5) node[left] {${#6}$};
\draw (0.5,-0.5)  node[right] {${#5}$};
\draw (-0.5,0.5) node[left] {${#3}$};
\draw  (0.5,0.5) node[right] {${#4}$};
\draw (-0.5,-0.5) node[bblob] {};
\draw (0.5,-0.5)  node[bblob] {};
\draw (-0.5,0.5) node[bblob] {};
\draw  (0.5,0.5) node[bblob] {};
\draw[aline=0.9] (0,0.5) -- ( (0,-0.5);
\draw[blue,wavy=0.2] (0.5,0) -- (-0.5,0);
\end{scope}
}

\def\sosplaq(#1,#2){
\begin{scope}[shift={(#1,#2)}]
\draw[dgreen] (-0.5,-0.5) rectangle ++(1,1);
\draw(0.5,0.5) node[bblob] {};
\draw(0.5,-0.5) node[bblob] {};
\draw(-0.5,0.5) node[bblob] {};
\draw(-0.5,-0.5) node[bblob] {};
\draw[dgreen] (-0.5,0.2) -- (-0.3,0.5);
\end{scope}
}

\def\udashsosplaq(#1,#2){
\begin{scope}[shift={(#1,#2)}]
\draw[dgreen,dashed] (-0.5,0.5) -- (0.5,0.5) -- (0.5,-0.5);
\draw[dgreen] (-0.5,0.5)-- (-0.5,-0.5) -- (0.5,-0.5);
\draw(0.5,0.5) node[bblob] {};
\draw(0.5,-0.5) node[bblob] {};
\draw(-0.5,0.5) node[bblob] {};
\draw(-0.5,-0.5) node[bblob] {};
\draw[dgreen] (-0.5,0.2) -- (-0.3,0.5);
\end{scope}
}

\def\lsosplaq(#1,#2){
\begin{scope}[shift={(#1,#2)}]
\draw[dgreen] (-0.5,-0.5) rectangle (0.5,0.2) ++(1,1);
\draw(0.5,0.2) node[bblob] {};
\draw(0.5,-0.5) node[bblob] {};
\draw(-0.5,0.2) node[bblob] {};
\draw(-0.5,-0.5) node[bblob] {};
\draw[dgreen] (-0.5,-0.1) -- (-0.3,0.2);
\end{scope}
}

\def\usosplaq(#1,#2){
\begin{scope}[shift={(#1,#2)}]
\draw[dgreen] (-0.5,-0.2) rectangle (0.5,0.5) ++(1,1);
\draw(0.5,0.5) node[bblob] {};
\draw(0.5,-0.2) node[bblob] {};
\draw(-0.5,0.5) node[bblob] {};
\draw(-0.5,-0.2) node[bblob] {};
\draw[dgreen] (-0.5,0.2) -- (-0.3,0.5);
\end{scope}
}
\def\lslantsosplaq(#1,#2){
\begin{scope}[shift={(#1,#2)}]
\draw[dgreen] (-0.5,-0.5) -- (0.5,-0.5) -- (0.5,0.2) -- (-0.5,0.5) -- (-0.5,-0.5) ++(1,1);
\draw(0.5,0.2) node[bblob] {};
\draw(0.5,-0.5) node[bblob] {};
\draw(-0.5,0.5) node[bblob] {};
\draw(-0.5,-0.5) node[bblob] {};
\draw[dgreen] (-0.5,0.2) -- (-0.3,0.4);
\end{scope}
}

\def\uslantsosplaq(#1,#2){
\begin{scope}[shift={(#1,#2)}]
\draw[dgreen] (-0.5,-0.5) -- (0.5,-0.2) -- (0.5,0.5) -- (-0.5,0.5) -- (-0.5,-0.5) ++(1,1);
\draw(0.5,0.5) node[bblob] {};
\draw(0.5,-0.2) node[bblob] {};
\draw(-0.5,0.5) node[bblob] {};
\draw(-0.5,-0.5) node[bblob] {};
\draw[dgreen] (-0.5,0.2) -- (-0.3,0.5);
\end{scope}
}

\def\luslantsosplaq(#1,#2){
\begin{scope}[shift={(#1,#2)}]
\draw[dgreen] (-0.5,-0.2) -- (0.5,-0.5) -- (0.5,0.5) -- (-0.5,0.5) -- (-0.5,-0.2) ++(1,1);
\draw(0.5,0.5) node[bblob] {};
\draw(0.5,-0.5) node[bblob] {};
\draw(-0.5,0.5) node[bblob] {};
\draw(-0.5,-0.2) node[bblob] {};
\draw[dgreen] (-0.5,0.2) -- (-0.3,0.5);
\end{scope}
}

\def\rcasosplaq(#1,#2){
\begin{scope}[shift={(#1,#2)}]
\draw[dgreen] (-0.75,-0.5) -- (0.5,-0.5) -- (0.5,0.5) -- (-0.5,0.5) -- (-0.75,-0.5) ++(1,1);
\draw(0.5,0.5) node[bblob] {};
\draw(0.5,-0.5) node[bblob] {};
\draw(-0.5,0.5) node[bblob] {};
\draw(-0.75,-0.5) node[bblob] {};
\draw[dgreen] (-0.6,0.15) -- (-0.3,0.5);
\end{scope}
}

\def\rcbsosplaq(#1,#2){
\begin{scope}[shift={(#1,#2)}]
\draw[dgreen] (-0.5,-0.5) -- (0.75,-0.5) -- (0.5,0.5) -- (-0.5,0.5) -- (-0.5,-0.5) ++(1,1);
\draw(0.5,0.5) node[bblob] {};
\draw(0.75,-0.5) node[bblob] {};
\draw(-0.5,0.5) node[bblob] {};
\draw(-0.5,-0.5) node[bblob] {};
\draw[dgreen] (-0.5,0.2) -- (-0.3,0.5);
\end{scope}
}

\def\rwidelsosplaq(#1,#2){
\begin{scope}[shift={(#1,#2)}]
\draw[dgreen] (-0.5,-0.2) rectangle (1,0.5) ++(1,1);
\draw(1,0.5) node[bblob] {};
\draw(1,-0.2) node[bblob] {};
\draw(-0.5,0.5) node[bblob] {};
\draw(-0.5,-0.2) node[bblob] {};
\draw[dgreen] (-0.5,0.2) -- (-0.3,0.5);
\end{scope}
}

\def\rwidesosplaq(#1,#2){
\begin{scope}[shift={(#1,#2)}]
\draw[dgreen] (-0.5,-0.5) rectangle (1,0.5) ++(1,1);
\draw(1,0.5) node[bblob] {};
\draw(1,-0.5) node[bblob] {};
\draw(-0.5,0.5) node[bblob] {};
\draw(-0.5,-0.5) node[bblob] {};
\draw[dgreen] (-0.5,0.2) -- (-0.3,0.5);
\end{scope}
}

\def\tallsosplaq(#1,#2){
\begin{scope}[shift={(#1,#2)}]
\draw[dgreen] (-0.5,-0.5) rectangle (0.5,0.8) ++(1,1);
\draw(0.5,0.8) node[bblob] {};
\draw(0.5,-0.5) node[bblob] {};
\draw(-0.5,0.8) node[bblob] {};
\draw(-0.5,-0.5) node[bblob] {};
\draw[dgreen] (-0.5,0.5) -- (-0.3,0.8);
\end{scope}
}

\pgfdeclaremetadecoration{prewavy}{zigzag}{
  \state{zigzag}[width=\pgfmetadecorationsegmentamplitude*\pgfmetadecoratedpathlength - 0.5*\pgfmetadecorationsegmentlength ,
                  next state=straight] {
        \decoration{zigzag}
      }
    \state{straight}[width=\pgfmetadecorationsegmentlength,
                   next state=final] {
        \decoration{curveto}
    }
    \state{final}{
        \decoration{zigzag}
        \beforedecoration{\pgfpathmoveto{\pgfpointmetadecoratedpathfirst}}
    }
}

\tikzset{
  wavy/.style={
    decorate,
    decoration={
      prewavy,
      meta-amplitude=#1,
      meta-segment length=0.3cm,
      amplitude=1.5pt, 
      segment length=6pt 
},
    postaction={decorate,ultra thick,decoration={markings,mark = at position #1 with {\arrow{>}}}}        
  },
  wavy/.default=0.5
}
\tikzset{oper/.style={rectangle,fill,inner sep=2.5pt}}
\tikzset{arr/.style={postaction={decorate,thick,decoration={markings,mark = at position #1 with {\arrow{>}}}}}}

\tikzset{
  abline/.style={blue,dashed,
    decorate,
    decoration={
      meta-amplitude=#1,
      meta-segment length=0.3cm,
},
    postaction={decorate,ultra thick,decoration={markings,mark = at position #1 with {\arrow{>}}}}        
  },
  abline/.default=0.5
}

\tikzset{
  aline/.style={
    decorate,
    decoration={
      meta-amplitude=#1,
      meta-segment length=0.3cm,
},
    postaction={decorate,ultra thick,decoration={markings,mark = at position #1 with {\arrow{>}}}}        
  },
  aline/.default=0.5
}

\tikzset{
  agline/.style={green!50!black,thick,
    decorate,
    decoration={
      meta-amplitude=#1,
      meta-segment length=0.3cm,
},
    postaction={decorate,ultra thick,decoration={markings,mark = at position #1 with {\arrow{>}}}}        
  },
  agline/.default=0.5
}

\tikzset{
  gline/.style={color=green!50!black,thick}}

\newtheorem{theorem}{Theorem}

\newtheorem{proposition}[theorem]{Proposition}

\makeatletter

\@addtoreset{equation}{section}
\makeatother

\newcommand{\be}{\begin{eqnarray}}
  \newcommand{\ee}{\end{eqnarray}}
\newcommand{\ben}{\begin{eqnarray*}}
  \newcommand{\een}{\end{eqnarray*}}
\newcommand{\bec}{\begin{equation}\begin{array}{lll}}
    \newcommand{\eec}{\end{array}\end{equation}}

\newcommand{\C}{\mathbb{C}}

\newcommand{\cC}{{\mathcal C}}

\newcommand{\cA}{{\mathcal A}}
\newcommand{\cB}{{\mathcal B}}
\newcommand{\cD}{{\mathcal D}}

\newcommand{\gl}{\lambda}

\newcommand{\ep}{\varepsilon}

\newcommand{\id}{\mathbb{I}}

\newcommand{\ot}{\otimes}
\newcommand{\sli}{\sum\limits}

\newcommand{\nn}{\nonumber}

\newcommand{\Ker}{\hbox{Ker}}
\newcommand{\Ima}{\hbox{Im}}

\newcommand{\ra}{\rightarrow}
\newcommand{\vac}{\omega^{(N)}}

\title{Lattice Supersymmetry in the Open XXZ Model: An Algebraic Bethe Ansatz Analysis}

\author{Robert Weston\thanks{\tt R.A.Weston@hw.ac.uk}\;}
\author{Junye Yang\thanks{\tt sttegunius@gmail.com }}
\affil{Department of Mathematics, Heriot Watt University, Edinburgh EH14 4AS, UK,
  and Maxwell Institute for Mathematical Sciences, Edinburgh, UK}
\date{September 2017}

\begin{document}
\maketitle

\bibliographystyle{unsrt}
\begin{abstract}
  \noindent 
  We reconsider the open XXZ chain of Yang and Fendley. This model possesses a lattice supersymmetry which changes the length of the chain by one site.  We perform an algebraic Bethe ansatz analysis of the model and derive the commutation relations of the lattice SUSY operators with the four elements of the open-chain monodromy matrix. Hence we give the action of the SUSY operator on off-shell and on-shell Bethe states. We show that this action generally takes one on-shell Bethe eigenstate to another. The exception is that a zero-energy vacuum state will be a SUSY singlet. The SUSY pairings of Bethe roots we obtain are analogous to those found previously for closed chains by Fendley and Hagendorf by analysing  the Bethe equations.
\end{abstract}

\nopagebreak

\section{Introduction}
In this paper we revisit the lattice supersymmetry (SUSY) of the open XXZ chain with the following quantum Hamiltonian:
\be H^{(N)}=-\frac{1}{2} \sli_{i=1}^{N-1} \left(\sigma^x_i \sigma^x_{i+1} + \sigma^y_i \sigma^y_{i+1} -\frac{1}{2} \sigma^z_i \sigma^z_{i+1}\right)-\frac{1}{4} (\sigma_1^z+\sigma_N^z)+\left(\frac{3N-1}{4}\right) \id .\label{eq:XXZ1}\ee
This Hamiltonian acts on the vector space $V^{\ot N}$ where $V\simeq \C^2$.
The surprising observation of \cite{MR2025875,MR2089982} was that this Hamiltonian can be expressed in terms of operators $Q^{(N)}: V^{\ot N}\ra V^{\ot N-1}$ and $Q^{(N)\dagger}:V^{(N-1)}\ra V^{(N)}$ as 
\be H^{(N)}= Q^{(N)\dag} Q^{(N)}+ Q^{(N+1)}  Q^{(N+1)\dag}.\label{eq:iHSUSYdef}\ee
Moreover these operators $Q^{(N)}$ and $Q^{(N)\dag}$ are nilpotent, i.e., 
\be  Q^{(N-1)} Q^{(N)}=0,\quad Q^{(N+1)\dagger} Q^{(N)\dagger}=0,\label{eq:inilpotency}\ee
from which is follows that 
\be H^{(N-1)} Q^{(N)}=Q^{(N)} H^{(N)},\quad  H^{(N)} Q^{(N)\dag}=Q^{(N)\dag} H^{(N-1)}.\label{eq:iSUSY}\ee
The properties \eqref{eq:iHSUSYdef}--\eqref{eq:iSUSY} allowed the authors of \cite{MR2025875,MR2089982} to characterise $H^{(N)}$ as a supersymmetric Hamiltonian and $Q^{(N)}$, $Q^{(N)\dag}$ as SUSY operators. The novel aspect of this SUSY is that it changes the length of the lattice by one unit.

 The Hamiltonian \eqref{eq:XXZ1} and its lattice SUSY are of interest for various reasons. Firstly, the SUSY allows us to characterise a vacuum eigenstate as a  SUSY singlet with zero energy (the simple reasons for this are explained nicely in the recent paper \cite{MR3635719}), and to organise the remaining spectrum into SUSY doublets. Secondly, the summation term in \eqref{eq:XXZ1} is that of the XXZ model at its `combinatorial point' with anisotropy parameter $\Delta=-\frac{1}{2}$. 
The ground state of the closed XXZ chain at this point has been much studied \cite{zinn2010six} and it displays a rich combinatorial structure that is realised most famously in the Razumov-Stroganov theorem \cite{MR2077318,MR2771600}. Recently, interesting combinatorial structure has also been observed  in the ground state of the open Hamiltonian \eqref{eq:XXZ1} \cite{MR3635719}. Finally, the XXZ model is of course a paradigmatic model of quantum integrability, with applications ranging from laboratory 
condensed-matter systems \cite{PhysRevLett.106.217203} to the AdS/CFT correspondence \cite{Beisert:2010jr} to non-equilibrium stochastic models \cite{derrida1993exact}. 

The type of lattice SUSY displayed by \eqref{eq:XXZ1} has now been found and studied, primarily by Fendley and Hagendorf and collaborators, in a growing range of open and closed quantum spin chains \cite{MR2725542,MR2903041,MR3024149, MR3314720, MR3514161,  MR3635719,meidinger2014dynamic}. In particular,  the manifestation of SUSY as pairings of  sets of roots of the Bethe Equations was observed and studied for the
 closed 8-vertex model along its combinatorial line in \cite{MR2903041}. Interesting observations regarding the connection between lattice SUSY and Bethe roots in other models can be found in \cite{MR3314720,MR3280007}.

In this paper we carry out an algebraic Bethe ansatz analysis of the original open Hamiltonian \eqref{eq:XXZ1}.
Following Sklyanin \cite{MR953215}, we consider the components $\cA^{N)}(u)$, $\cB^{(N)}(u)$, $\cC^{(N)}(u)$, and $\cD^{(N)}(u)$ of the double-row monodromy matrix.   The main result of this paper is Theorem \ref{thm1} in Section 4, which expresses the commutation relations of the SUSY operator $Q^{(N)}$ with $\cA^{N)}(u)$, $\cB^{(N)}(u)$, $\cC^{(N)}(u)$, and $\cD^{(N)}(u)$. As such it gives the connection between the SUSY operators and the underlying reflection algebra \cite{MR953215,MR1193836}.  
The action of $Q^{(N)}$ on both off and on-shell Bethe states is then a simple corollary of our main theorem. 
Recall that `off-shell' Bethe states are constructed as 
\ben \cB^{(N)}(\gl_1) \cB^{(N)}(\gl_2) \cdots \cB^{(N)}(\gl_m) \Omega^{(N)},\een
where $\Omega^{(N)}$ is the pseudo-vacuum with all spins up and $\gl_i$ are generic complex numbers. It follows from our Theorem \ref{thm1} that
\be
&&\hspace*{-5mm}Q^{(N)} \,\cB^{(N)}(\lambda_1)\,\cB^{(N)}(\lambda_2) \cdots \cB^{(N)}(\lambda_m) \, \Omega^{(N)}\nn\\[1mm]
&&
\propto \cB^{(N-1)}(\lambda_1)\,\cB^{(N-1)}(\lambda_2) \cdots \cB^{(N-1)}(\lambda_m) \,\cB^{(N-1)}(2\pi i/3) \, \Omega^{(N-1)},\quad\hbox{ if }
\gl_i\neq 2\pi i/3\hbox { for all }i, \nn\\
&&= 0 \quad \hbox{otherwise}.\label{eq:qreln}
\ee
Moreover, we find that if $\{\gl_1,\gl_2,\cdots, \gl_m\}$ obey the Bethe equations \eqref{eq:Beqns} for the length $N$ chain (that is they are Bethe roots, such that the resulting Bethe state is `on-shell') with $\gl_i\neq2\pi i/3$ for all $i$, then $\{\gl_1,\gl_2,\cdots, \gl_m,2\pi i/3\}$ obey the Bethe equations for the length $N-1$ chain. 
The result \eqref{eq:qreln} complements and extends the previous results concerning the existence of SUSY pairings of Bethe roots observed for closed  chains \cite{MR2903041}, and also explains such pairings directly in terms of the action of the lattice SUSY on Bethe states. 

The layout of this paper is as follows: in Section 2 we review the construction given in \cite{MR2089982} of the Hamiltonian \eqref{eq:XXZ1} in terms of SUSY operators. Then we briefly describe the necessary aspects of the algebraic Bethe ansatz \cite{MR953215} approach to open chains in Section 3. In Section 4 we present our main results regarding the commutation relations of the SUSY operators and the components of the double-row monodromy matrix, and hence give the action of the SUSY operators on off-shell and on-shell Bethe states.  Finally, we discuss the potential uses of our results in Section 5.

\section{Lattice SUSY of the Open XXZ Model}
In this section, we recall the supersymmetric picture of the open XXZ Hamiltonian $H^{(N)}$ defined by Equation \eqref{eq:XXZ1}.
It was discovered in \cite{MR2025875, MR2089982} that this model possesses a lattice SUSY that relates $H^{(N)}$ and $H^{(N-1)}$. More precisely, let us define the vector space $V=\C v_+\oplus \C v_-$, such that $H^{(N)}: V^{\ot N} \ra V^{\ot N}$.  Then one can define operators
 $q: V\ot V \ra V$ and $q^\dagger: V \ra V\ot V$ by 
\be q(v_{\ep_1}\ot v_{\ep_2})= \delta_{\ep_1,+}\delta_{\ep_2,+} \, v_-,\quad q^\dagger \, v_\ep = \delta_{\ep,-} v_+\ot v_+,\label{eq:qdef}\ee The operators $q$ and $q^\dagger$ obeys the associativity and coassociativity conditions 
\be q(q\ot \id) = q(\id \ot q),\quad (q^\dag\ot \id) q^\dag = (\id \ot q^\dag) q^\dag.\label{eq:coassoc}\ee
We then define operators $Q^{(N)}: V^{\ot N} \ra V^{\ot N-1}$  and $Q^{(N)\dag}: V^{\ot N-1}\ra V^{\ot N}$ by \be
Q^{(N)}=\sli_{i=1}^{N-1} (-1)^{i+1} q_{i,i+1},\quad Q^{(N)\dag}=\sli_{i}^{N-1} (-1)^{i+1} q^\dag_i,\label{eq:Qdef}\ee
where subscripts indicate the copies of $V$ in the $V^{\ot N}$ on which operators act non-trivially.\footnote{Note that we swap notations $q\leftrightarrow q\dagger$ and $Q
 \leftrightarrow Q^\dagger$ compared to the work of Fendley, Hagendorf and collaborators \cite{MR2089982,MR3024149} - we like to think of $Q^\dag$ as creating a lattice site.}

It follows from Equation \eqref{eq:coassoc} that $Q^{(N)}$ and $Q^{(N)\dagger}$ are nilpotent: 
\be  Q^{(N-1)} Q^{(N)}=0,\quad Q^{(N+1)\dag} Q^{(N)\dag}=0,\label{eq:nilpotency}\ee
and it is a relatively simple calculation to show that we can write 
 \be H^{(N)}= Q^{(N)\dag} Q^{(N)}+ Q^{(N+1)}  Q^{(N+1)\dag}: V^{\ot N} \ra V^{\ot N}.\label{eq:HSUSYdef}\ee
Then it follows from the nilpotency conditions \eqref{eq:nilpotency} that we have 
\be H^{(N-1)} Q^{(N)}=Q^{(N)} H^{(N)},\quad  H^{(N)} Q^{(N)\dag}=Q^{(N)\dag} H^{(N-1)}.\label{eq:SUSY}\ee
As mentioned in the Introduction, the properties \eqref{eq:nilpotency}--\eqref{eq:SUSY} allow us to regard $H^{(N)}$ as a SUSY Hamiltonian (for a simple introduction to SUSY in quantum mechanics see for example \cite{MR1849169}).

\section{The Algebraic Bethe Ansatz Analysis of the SUSY XXZ Chain}
In this section we recall the algebraic Bethe ansatz approach to the open XXZ chain formulated by Sklyanin \cite{MR953215}. We start from an R-matrix $R(u):V\ot V \rightarrow V \ot V$ and K-matrices $K^{\pm}(u):V\rightarrow V$
given by\footnote{We specialise the boundary parameters of  \cite{MR953215} to be $\xi_+=\xi-=-\eta$.}
\be R(u)=\begin{pmatrix} a(u) & 0 & 0 & 0 \\ 0  & b(u) & c(u) & 0\\ 0  & c(u) & b(u) & 0\\
0 & 0 & 0 & a(u) \end{pmatrix}, \quad K^-(u)=\begin{pmatrix} d(u) & 0 \\ 0 & -a(u)\end{pmatrix}, \quad K^+(u)=K^-(u+\eta), \label{eq:RKdef},\ee
where \be
a(u)=\frac{\sinh(u+\eta)}{\sinh(\eta)},\quad b(u)=\frac{\sinh(u)}{\sinh(\eta)},\quad c(u)=1,\quad d(u)=\frac{\sinh(u-\eta)}{\sinh(\eta)}.
\label{eq:addef}\ee\newpage
\noindent The matrices $R(u)$ and $K^\pm(u)$ obey the Yang-Baxter, boundary Yang-Baxter and unitarity/crossing-unitarity relations which can be found in 
 \cite{MR953215}.
 
In terms of these operators, we can construct a monodromy matrix 
\be U^{(N)}(u)= \overline{T}^{(N)}(u) K_0^-(u)\, T^{(N)}(u)= \begin{pmatrix} \mathcal{A}^{(N)}(u) & \mathcal{B}^{(N)}(u)  \\ \mathcal{C}^{(N)}(u) &  \mathcal{D}^{(N)}(u)  \end{pmatrix}\in \hbox{End}(V_0\ot V_N \ot \cdots \ot V_2\ot V_1) ,\nn\\ \label{eq:Udef} \\[-12mm]\nn\ee
\ben \hspace*{-16mm}\hbox{with}\quad T^{(N)}(u)= R_{01}(u)R_{02}(u) \cdots R_{0N}(u), \quad \overline{T}^{(N)}(u)= R_{N0}(u) \cdots R_{20}(u) R_{10}(u). \een
Here and elsewhere the subscripts denote the relevant labels\footnote{Our seemingly idiosyncratic choice of $0,N,.\cdots,2,1$ for the ordering of labels is useful for the inductive proof of Theorem 1} of the copies of $V$.

We define a transfer matrix $t^{(N)}(u)$ by
\ben t^{(N)}(u)= \hbox{Tr}_{V_0} \left( K^+_0(u)U^{(N)}(u)\right)=\frac{\sinh(u)}{\sinh(\eta)} \cA^{(N)}(u) - \frac{\sinh(u+2\eta)}{\sinh(\eta)}\cD^{(N)}(u),\een
for which it is shown in \cite{MR953215} that the standard quantum integrability condition holds. Namely,
\ben [t^{(N)}(u),t^{(N)}(v)]=0.\een
Taking the derivative of the transfer matrix gives
\ben t^{(N)\prime}(0)&=&\hbox{Tr}_{V_0} \left( K^{+\prime}_0(0))\, K^{-}_1(0)\right)\\&&+\hbox{Tr}_{V_0} \left( K^+_0(0)\right) 
\left[K^{-\prime}_1(0) + 2\sli_{i=1}^{N-1} h_{i+1,i} K^-_1(0) \right] +2\,\hbox{Tr}_{V_0} \left( K^{+}_0(0) h_{0N}\right)  K^-_1(0),\\\hbox{where      }\quad \quad
h&=&PR'(0)=\frac{1}{2\sinh(\eta)} \left(\sigma^x\ot \sigma^x+\sigma^y\ot \sigma^y + \cosh(\eta) \,\sigma^z \ot \sigma^z\right)+ \frac{\coth(\eta)}{2} \,\id.\een
Inserting the expressions \eqref{eq:RKdef}--\eqref{eq:addef} for $K^\pm(u)$ gives
\ben  
t^{(N)\prime}(0)= \frac{2\cosh(2\eta)}{\sinh(\eta)} \id + 4 \cosh(\eta) \sli_{i=1}^{N-1} h_{i+1,i} -\frac{2\cosh^2(\eta)}{\sinh(\eta)} (\sigma_1^z+\sigma_N^z),\een
and specialising to $\eta=\frac{2\pi i}{3}$, which we shall do for the rest of this paper,  we find 
\be t^{(N)\prime}(0)= \frac{-4i}{\sqrt{3}}
 \left[H^{(N)}+ \frac{(1-N)}{2} \id\right] ,\label{eq:tHcon}\ee 
where $H^{(N)}$ is as defined in Equation \eqref{eq:XXZ1}. It is useful for later purposes to note that after this $\eta=\frac{2\pi i}{3}$ specialisation we have the following identities for the functions \eqref{eq:addef} involved in the definitions of the $R$ and $K$ matrices:
\begin{equation} a(u)+b(u)+d(u)=0, \quad a(u)^2+b(u)^2-c(u)^2+a(u)b(u) =0.\label{eq:ids}\end{equation}

The algebraic Bethe ansatz for our open chain starts from the ansatz that at least some common eigenvectors of $t^{(N)}(u)$ and $H^{(N)}$ can be written in the form 
\be \cB^{(N)}(\lambda_1) \cB^{(N)}(\lambda_2) \cdots \cB^{(N)}(\lambda_m) \, \Omega^{(N)},\quad m\leq N,\label{eq:Bstate}\ee
where $\Omega^{(N)}$ is the `pseudovacuum' defined by $\Omega^{(N)}=v_+\ot v_+ \ot \cdots \ot v_+ \in V^{\ot N}$.
Demanding that the state \eqref{eq:Bstate} is an eigenvalue of $t^{(N)}(u)$ yields the Bethe equations\footnote{Our Bethe Equations differ slightly
from Sklyanin's\cite{MR953215}. The reason is that we use spectral parameters $u$ and $0$ on our auxiliary and quantum spaces, whereas Sklyanin uses $u-\eta/2$ and $-\eta/2$.}
\be -\sinh(2\gl_j) \frac{\Delta_+^{(N)}(\gl_j)}{\Delta_-^{(N)}(\gl_j)}= \prod_{k\neq j} \frac{\sinh(\gl_j-\gl_k+\eta)\sinh(\gl_j+\gl_k+2\eta) }{\sinh(\gl_j-\gl_k-\eta)\sinh(\gl_j+\gl_k)}\label{eq:Beqns} \ee
where $\Delta_\pm^{(N)}(\lambda)$ are related to the eigenvalues of $\Omega^{(N)}$ with respect to $\cA^{(N)}(\lambda)$ and $\cD^{(N)}(\lambda)$ thus:
\be\cA^{(N)}(\lambda) \Omega^{(N)}= \Delta_+^{(N)}(\lambda) \Omega^{(N)}, \quad \left[\cD^{(N)}(\lambda) \sinh(2\lambda+\eta) -\cA^{(N)}(\lambda)\sinh(\eta)\right] \Omega^{(N)}= \Delta_-^{(N)}(\lambda) \Omega^{(N)}.\nn\\ \label{eq:Deltadef} \ee
Explicit expressions for $\Delta_\pm^{(N)}(\lambda)$ are derived in Section 4 (see Equation \eqref{eq:dpm}).
The corresponding eigenvalue of the transfer matrix $t^{(N)}(u)$  (which we denote by $\tau^{(N)}(u;\lambda_1,\lambda_2,\cdots,\lambda_m)$) is then
\be \tau^{(N)}(u;\lambda_1,\lambda_2,\cdots,\lambda_m) &=& \frac{\sinh(2u-\eta)\sin(u-\eta) \Delta_+^{(N)}(u) }{\sinh(2u+\eta)} \prod\limits_{j=1}^m 
\frac{\sinh(u-\lambda_j-\eta)\sinh(u+\lambda_j)}{\sinh(u-\lambda_j)\sinh(u+\lambda_j+\eta)}\nonumber\\
&-&\frac{\sinh(u-\eta) \Delta_-^{(N)}(u)}{\sinh(2u+\eta)} \prod\limits_{j=1}^m 
\frac{\sinh(u-\lambda_j+\eta)\sinh(u+\lambda_j-\eta)}{\sinh(u-\lambda_j)\sinh(u+\lambda_j+\eta)}.
\label{eq:teigenval}\ee

Following the standard convention, we refer to states of the form \eqref{eq:Bstate}, in which the $\lambda_i$ take general complex values as `off-shell' Bethe states. If the  $\lambda_i$ are restricted by \eqref{eq:Beqns}, then we refer to the states as `on-shell' Bethe states.

\section{The Action of SUSY Operators on Bethe States}
In this section, we consider the action of the SUSY generator $Q^{(N)}$ on first off-shell and then on-shell Bethe states. That is, we consider the state \ben
Q^{(N)} \cB^{(N)}(\lambda_1)\cB^{(N)}(\lambda_2) \cdots \cB^{(N)}(\lambda_m) \, \Omega^{(N)}.\een The key result is Theorem 1 below. 
To simplify the proof of this theorem it useful to introduce some additional notation. 
Firstly, we define operators $E^{\ep_1}_{\ep_2}:V\ra V$, $E^{\ep_1}: V\ra \C$ and  $E^{\ep_1\ep_2}_{\;\;\ep_3}: V\ot V \ra V$, with $\ep_i\in\{+,-\}$, by
\ben E^{\ep_1}_{\ep_2} v_{\ep'_1} = \delta_{\ep_1,\ep'_1} v_{\ep_2},\quad E^{\ep_1} v_{\ep'_1}=\delta_{\ep_1,\ep'_1},\quad E^{\ep_1\ep_2}_{\;\;\ep_3} (v_{\ep'_1}\ot v_{\ep'_2}) = \delta_{\ep_1,\ep'_1}\delta_{\ep_2,\ep'_2}\, v_{\ep_3}.\een
Note that the operator $q$ defined by \eqref{eq:qdef} is then given by $q=E^{++}_-$. Clearly we have \be E^{\ep_1} \ot E^{\ep_2}_{\ep_3}= E^{\ep_1\ep_2}_{\;\;\ep_3}= E^{\ep_1}_{\ep_3} \ot E^{\ep_2}.\label{eq:Eprop}\ee
Secondly, we make use of the standard pictorial representations
\ben 
\begin{tikzpicture}[scale=0.75]
\draw(-2,-1) node[left] {$R(u_1-u_2)^{\ep_1\ep_2}_{\ep'_1\ep'_2}\;=$};
 \draw[aline=0.9](-1,-0.5) -- (1,-0.5);
  \draw[aline=0.9](0,0.5) -- (0,-1.5);
\draw(0,0.5) node[above] {$\ep_2$};
\draw(-1,-0.5) node[left] {$\ep_1$};
\draw(1,-0.5) node[right] {$\ep'_1$};
\draw(0,-1.5) node[below] {$\ep'_2$};
\draw(-1/2,-0.5) node[below] {$u_1$};
\draw(0.3,0) node {$u_2$};
\end{tikzpicture}
\quad
\begin{tikzpicture}[scale=0.75]
\draw(-1,0) node[left] {$  K^-(u)^{\ep}_{\ep'} \;=$}; \draw (0,1) arc (90:-90:1);
\draw(1,0)  node[bblob]{};
\draw[aline=1] (-0.5,1) --(0,1);
\draw[aline=0] (0,-1) --(-0.5,-1);
\draw(-0.5,1) node[left] {$\ep$};
\draw(-0.5,-1) node[left] {$\ep'$};
\draw(1.3,0.7) node {$u$};
\draw(1.2,-0.7) node {$-u$};
\end{tikzpicture}
\quad
\begin{tikzpicture}[scale=0.75]
\draw(-1,0) node[left] {$  K^+(u)^{\ep}_{\ep'} \;=$}; 
\draw (0.5,1) arc (90:270:1);
\draw(-0.5,0)  node[bblob]{};
\draw[aline=0] (0.5,1) --(1,1);
\draw[aline=1] (1,-1) --(0.5,-1);
\draw(1,1) node[right] {$\ep'$};
\draw(1,-1) node[right] {$\ep$};
\draw(-0.9,0.7) node {$u$};
\draw(-0.7,-0.7) node {$-u$};
\end{tikzpicture}
\een
This leads to the following convenient pictorial representation for the operators $U^{(N)}(u)$ and $t^{(N)}(u)$ defined in Section 3:

\be
\begin{tikzpicture}[scale=0.5]
  \draw[aline=0.95](-7,1) -- (1,1);
  \draw[aline=0.05](1,-1) -- (-7,-1);
\foreach\x in {-6,-4,-2,0}{
    \draw[aline=0.9](\x,2) -- (\x,-2);
  }
 \draw (1,1) arc (90:-90:1);
\draw(2,0)  node[bblob]{};
 \draw(1,1.1) node[above] {$u$};
\draw(1,-1.1) node[below] {$-u$};
\draw(-10,0) node[left] {$U^{(N)}(u)=$};
\draw(-6,2.5) node[above] {$N$};
\draw(-4,2.5) node[above] {$\cdots$};
\draw(-2,2.5) node[above] {$2$};
\draw(-0,2.5) node[above] {$1$};
\draw(-6,0) node[right] {$0$};
\draw(-4,0) node[right] {$0$};
\draw(-2,0) node[right] {$0$};\draw(-0,0) node[right] {$0$};
\end{tikzpicture}\label{eq:Upic}\\[4mm]
\begin{tikzpicture}[scale=0.5]
  \draw[aline=0.95](-7,1) -- (1,1);
  \draw[aline=0.05](1,-1) -- (-7,-1);
\foreach\x in {-6,-4,-2,0}{
    \draw[aline=0.9](\x,2) -- (\x,-2);
  }
 \draw (1,1) arc (90:-90:1);
  \draw (-7,-1) arc (-90:-270:1);
\draw(2,0)  node[bblob]{};
\draw(-8,0)  node[bblob]{};
 \draw(1,1.1) node[above] {$u$};
\draw(1,-1.1) node[below] {$-u$};
\draw(-10,0) node[left] {$t^{(N)}(u)=$};
\draw(-6,2.5) node[above] {$N$};
\draw(-4,2.5) node[above] {$\cdots$};
\draw(-2,2.5) node[above] {$2$};
\draw(-0,2.5) node[above] {$1$};
\draw(-6,0) node[right] {$0$};
\draw(-4,0) node[right] {$0$};
\draw(-2,0) node[right] {$0$};\draw(-0,0) node[right] {$0$};
\end{tikzpicture}
\ee
where $u$, $-u$ are $0$ are the additive spectral parameters, and
 the numbers $N, \cdots,2,1$ label the vertical quantum spaces.

We now present the key theorem regarding the commutation relation of $Q^{(N)}$ with the elements of the monodromy matrix $U^{(N)}(u)$. In the statement and proof of this theorem, we suppress function arguments and abbreviate $\cA^{(N)}=\cA^{(N)}(u),\cdots, \cD^{(N)}=\cD^{(N)}(u)$ and $a=a(u)$, \dots, $d=d(u)$ (where $a,b,c,d$ are defined in Equation \eqref{eq:addef}). 
\begin{theorem} \label{thm1} We have the following commutation relations for $N\geq 2$:
\be
Q^{(N)} \cA^{(N)}- d^2 \cA^{(N-1)} Q^{(N)}&\hspace*{-2mm}=&\hspace*{-2mm} (-1)^N c\, d\,  E^+ \ot \cB^{(N-1)} \label{eq:ft1},\\
Q^{(N)} \cB^{(N)}- d^2 \cB^{(N-1)} Q^{(N)} &\hspace*{-2mm}=&\hspace*{-2mm}0\label{eq:ft2},\\
\hspace*{-4mm}Q^{(N)} \cC^{(N)}- d^2 \cC^{(N-1)} Q^{(N)}&\hspace*{-2mm}=&\hspace*{-2mm}(-1)^N \left( a\, c\,  E^+\ot \cA^{(N-1)} +c^2\,  
E^-\ot \cB^{(N-1)} +  c\, d\,  E^+\ot \cD^{(N-1)} \right)\!,\qquad\label{eq:ft3}\\
Q^{(N)} \cD^{(N)}- d^2 \cD^{(N-1)} Q^{(N)} &\hspace*{-2mm}=&\hspace*{-2mm}(-1)^N  b\, c\,  E^+\ot \cB^{(N-1)}\label{eq:ft4}.
\ee 

\end{theorem}

\begin{proof} The proof is inductive with respect to $N$. The base step with $N=2$ is simple to show by explicit calculation using $Q^{(2)}=q=E^{++}_{\;-}$, the definition \eqref{eq:Udef} of $\cA^{(N)},\cdots,\cD^{(N)}$ and the the explicit form \eqref{eq:RKdef} 
of the $R$ and $K^-$ matrices. Let us consider the case of \eqref{eq:ft1} for illustrative purposes.
We find \ben Q^{(2)} \cA^{(2)}= a^4 d \, E^{++}_-,\quad \cA^{(1)} Q^{(2)} = (db^2-ac^2) E^{++}_-,\quad \cB^{(1)}=bc(d-a) E^+_-.\een
Hence we have 
\ben Q^{(2)} \cA^{(2)}- d^2 \cA^{(1)} Q^{(2)} - cd E^+\ot \cB^{(1)}= \Big(a^4 d -d^2(db^2-ac^2) -b c^2 d (d-a) \Big) E^{++}_{\;-}=0,\een
where the final equality follows from the identities \eqref{eq:ids}.

 In order to simplify the inductive step, it is useful to introduce the notation
\ben [X]^{(N)}= Q^{(N)} X^{(N)} - d^2 X^{(N-1)} Q^{(N)} : V^{\ot N}\ra V^{\ot N-1}\een for the shifted commutators occurring on the lhs of \eqref{eq:ft1}--\eqref{eq:ft4}. Then in terms of this notation, the main step of the proof is to show the following recursion relations:
\be \,&&\hspace*{-15mm}\cA^{(N+1)}= a^2  E^{+}_{+}\ot \cA^{(N)}  + b^2  E^{-}_{-}\ot \cA^{(N)}
+  ac E^-_+ \ot \cB^{(N)} 
+ ac  E^+_-\ot  \cC^{(N)}
+ c^2  E^-_- \ot \cD^{(N)}\!, \label{eq:srrA}\\
\, &&\hspace*{-15mm}\cB^{(N+1)}=bc E^+_- \ot\cA^{(N)} 
+ ab \, \id \ot \cB^{(N)}
+ bc E^+_- \ot  \cD^{(N)},\label{eq:srrB}\\
&&\hspace*{-15mm}\cC^{(N+1)}=bc  E^-_+ \ot \cA^{(N)} 
+ ab \,\id \ot \cC^{(N)}
+ bc E^-_+ \ot\cD^{(N)}, \label{eq:srrC}\\
&&\hspace*{-15mm}  \cD^{(N+1)}=c^2 E^+_+\ot \cA^{(N)}
+ ac E^-_+ \ot \cB^{(N)} +ac E^+_- \ot \cC^{(N)}
+ b^2  E^+_+ \ot \cD^{(N)}  + a^2  E^-_-  \ot \cD^{(N)}\!, 
\label{eq:srrD}
\ee
and 
\be
   \, [\cA]^{(N+1)}&=& a^2  E^{+}_{+}\ot [\cA]^{(N)}  + b^2  E^{-}_{-}\ot [\cA]^{(N)}
+  ac E^-_+ \ot [\cB]^{(N)} 
+ ac  E^+_-\ot  [\cC]^{(N)}
+ c^2  E^-_- \ot [\cD]^{(N)}\nn \\
&&\hspace*{-10mm}+(-1)^N \left\{(b^2\,d^2\,-a^4) \,  E^{++}_{\;-}\ot \cA^{(N-1)} + acd^2 \, E^{++}_{\;+}\ot \cB^{(N-1)}\right. \nn \\ && \hspace*{-10mm}\left. - a^2 b c \, E^{-+}_{\;-} \ot \cB^{(N-1)} 
- a^3 c  E^{+-}_{\;-}\ot  \cB^{(N-1)} + c^2 d^2 E^{++}_- \ot \cD^{(N-1)}\right\},\label{eq:rrA}\\[3mm]
  \, [\cB]^{(N+1)}&=&bc E^+_- \ot[\cA]^{(N)} 
+ ab \, \id \ot [\cB]^{(N)}
+ bc E^+_- \ot  [\cD]^{(N)}
+ (-1)^N ab(d^2-ab) \, E^{++}_{\;-} \ot \cB^{(N-1)},\nn\\ \label{eq:rrB} \\[3mm]
 \, [\cC]^{(N+1)}&=&bc  E^-_+ \ot [\cA]^{(N)} 
+ ab \,\id \ot [\cC]^{(N)}
+ bc E^-_+ \ot[\cD]^{(N)}
\nn \\
&&\hspace*{-10mm}+(-1)^N \left\{  bcd^2 \, E^{++}_{\;+} \ot A^{(N-1)} 
- b c(a^2+c^2) \, E^{-+}_{\;-} \ot A^{(N-1)}
- ab^2 c \, E^{+-}_{\;-}  \ot A^{(N-1)} \right.\nn \\&&\hspace*{-10mm}
\left. - 2abc^2\,  E^{--}_{\;-} \ot B^{(N-1)}
+ab(d^2-ab) \, E^{++}_{\;-} \ot C^{(N-1)}\right.\nn\\ &&\hspace*{-10mm}\left.
+ bcd^2 \,E^{++}_{\;+}  \ot D^{(N-1)}
-b^3 c  \, E^{-+}_{\;-} \ot D^{(N-1)} -ab^2c\, E^{+-}_{\;-}\ot D^{(N-1)}
 \right\},
\label{eq:rrC} \\[3mm]
 \, [\cD]^{(N+1)}&=&c^2 E^+_+\ot [\cA]^{(N)}
+ ac E^-_+ \ot [\cB]^{(N)} +ac E^+_- \ot [\cC]^{(N)}
+ b^2  E^+_+ \ot [\cD]^{(N)}  + a^2  E^-_-  \ot [\cD]^{(N)} 
\nn \\
&&\hspace*{-10mm}+(-1)^N \left\{  -c^2(a^2+b^2)  E^{++}_{\;-}\ot  \cA^{(N-1)}
+acd^2  E^{++}_{\;+} \ot B^{(N-1)} \right. \nn\\&&\hspace*{-10mm}\left.
-a^2 bc  E^{-+}_{\;-} \ot B^{(N-1)}
-ac(b^2+c^2) E^{+-}_{\;-} \ot  B^{(N-1)}
+(a^2d^2-b^4) E^{++}_{\;-} \ot D^{(N-1)}\right\}.
\label{eq:rrD}
\ee

Let us first consider the simplest relations \eqref{eq:srrB} and \eqref{eq:rrB}.
The recursion relation \eqref{eq:srrB} is most easily obtained by graphical arguments: the operator $\cB^{(N+1)}$ is represented by 
\be
\begin{tikzpicture}[scale=0.4]
  \draw[aline=0.4](-9,1) -- (1,1);
  \draw[aline=0.62](1,-1) -- (-9,-1);
\foreach\x in {-8,-6,-4,-2,0}{
    \draw[aline=0.9](\x,2) -- (\x,-2);
  }
 \draw (1,1) arc (90:-90:1);
\draw(2,0)  node[bblob]{};
\draw(-12,0) node[left] {$\cB^{(N+1)}=$};
\draw(-8,2.5) node[above] {$N+1$};
\draw(-6,2.5) node[above] {$N$};
\draw(-4,2.5) node[above] {$\cdots$};
\draw(-2,2.5) node[above] {$2$};
\draw(-0,2.5) node[above] {$1$};
\draw(-9,1) node[left] {$-$};
\draw(-9,-1) node[left] {$+$};
\end{tikzpicture}
\label{fig:Bpic}\ee
The non-zero-weight configurations of spins around the leftmost column of R-matrices, with the two leftmost spins fixed as shown in \eqref{fig:Bpic}, are
\ben\begin{tikzpicture}[scale=0.4]
  \draw[aline=0.8](-1,1) -- (1,1);
  \draw[aline=0.3](1,-1) -- (-1,-1);
\draw[aline=0.6](0,2) -- (0,-2);
\draw(1,1) node[right] {$+$};
\draw(1,-1) node[right] {$+$};
\draw(-1,1) node[left] {$-$};
\draw(-1,-1) node[left] {$+$};
\draw(0,2) node[above] {$+$};
\draw(0,-2) node[below] {$-$};
\draw(0,0) node[right] {$-$};
\draw(0,-3) node[below] {$bc$};
 \end{tikzpicture}
\qquad 
\begin{tikzpicture}[scale=0.4]
  \draw[aline=0.8](-1,1) -- (1,1);
  \draw[aline=0.3](1,-1) -- (-1,-1);
\draw[aline=0.6](0,2) -- (0,-2);
\draw(1,1) node[right] {$-$};
\draw(1,-1) node[right] {$+$};
\draw(-1,1) node[left] {$-$};
\draw(-1,-1) node[left] {$+$};
\draw(0,2) node[above] {$+$};
\draw(0,-2) node[below] {$+$};
\draw(0,0) node[right] {$+$};
\draw(0,-3) node[below] {$ab$};
 \end{tikzpicture}
\qquad 
\begin{tikzpicture}[scale=0.4]
  \draw[aline=0.8](-1,1) -- (1,1);
  \draw[aline=0.3](1,-1) -- (-1,-1);
\draw[aline=0.6](0,2) -- (0,-2);
\draw(1,1) node[right] {$-$};
\draw(1,-1) node[right] {$+$};
\draw(-1,1) node[left] {$-$};
\draw(-1,-1) node[left] {$+$};
\draw(0,2) node[above] {$-$};
\draw(0,-2) node[below] {$-$};
\draw(0,0) node[right] {$-$};
\draw(0,-3) node[below] {$ab$};
 \end{tikzpicture}
\qquad 
\begin{tikzpicture}[scale=0.4]
  \draw[aline=0.8](-1,1) -- (1,1);
  \draw[aline=0.3](1,-1) -- (-1,-1);
\draw[aline=0.6](0,2) -- (0,-2);
\draw(1,1) node[right] {$-$};
\draw(1,-1) node[right] {$-$};
\draw(-1,1) node[left] {$-$};
\draw(-1,-1) node[left] {$+$};
\draw(0,2) node[above] {$+$};
\draw(0,-2) node[below] {$-$};
\draw(0,0) node[right] {$+$};
\draw(0,-3) node[below] {$bc$};
 \end{tikzpicture}
\een
where the associated weights are shown below the configuration diagram. Expression \eqref{eq:srrB} follows directly. 
Next, we note that  $Q^{(N+1)} = \id \ot Q^{(N)} + (-1)^{N+1} q_{N+1,N}$, from which it follows that 
\be
[\cB]^{(N+1)}&=& bc \,E^+_- \ot [\cA]^{(N)}
+ ab \,\id \ot [\cB]^{(N)}
+ bc \, E^+_- \ot [\cD]^{(N)}\nn\\
&&+ (-1)^{N+1}\left\{q_{N+1,N} \, \cB^{(N+1)}-d^2 \cB^{(N)} \, q_{N+1,N}\right\}.
 \label{eq:rrBinter}\ee
If we represent $q=E^{++}_{\;-}$ by the following picture
\ben \begin{tikzpicture}[scale=0.5]
\draw(0,-2) node[left] {$q=$};
\draw[aline=0.6](1,-2) -- (1,-3);
\draw(2,-1) arc (-0:-180:1);
\end{tikzpicture}
\een
then we have 
\ben &&q_{N+1,N} \, \cB^{(N+1)}-d^2 \cB^{(N)} \, q_{N+1,N}=\\[2mm]&&
\begin{tikzpicture}[scale=0.4]
  \draw[aline=0.4](-9,1) -- (1,1);
  \draw[aline=0.62](1,-1) -- (-9,-1);
\foreach\x in {-4,-2,0}{
    \draw[aline=0.9](\x,3) -- (\x,-3);
  }
\draw[aline=0.9](-6,3) -- (-6,-1);
\draw[aline=0.9](-8,3) -- (-8,-1);
\draw[aline=0.6](-7,-2) -- (-7,-3);
 \draw (1,1) arc (90:-90:1);
\draw(2,0)  node[bblob]{};
\draw(-8,3.5) node[above] {$N+1$};
\draw(-6,3.5) node[above] {$N$};
\draw(-4,3.5) node[above] {$\cdots$};
\draw(-2,3.5) node[above] {$2$};
\draw(-0,3.5) node[above] {$1$};
\draw(-9,1) node[left] {$-$};
\draw(-9,-1) node[left] {$+$};
\draw(-8,-1) arc (-180:0:1);
\end{tikzpicture}\quad\quad
\begin{tikzpicture}[scale=0.4]
\draw(-11,0) node[left] {$-d^2$};
  \draw[aline=0.4](-9,1) -- (1,1);
  \draw[aline=0.62](1,-1) -- (-9,-1);
\foreach\x in {-4,-2,0}{
    \draw[aline=0.9](\x,3) -- (\x,-3);
  }
\draw(-7,2) -- (-7,-3);
\draw[aline=0.6](-7,-2) -- (-7,-3);
 \draw (1,1) arc (90:-90:1);
\draw(2,0)  node[bblob]{};
\draw(-8,3.5) node[above] {$N+1$};
\draw(-6,3.5) node[above] {$N$};
\draw(-4,3.5) node[above] {$\cdots$};
\draw(-2,3.5) node[above] {$2$};
\draw(-0,3.5) node[above] {$1$};
\draw(-9,1) node[left] {$-$};
\draw(-9,-1) node[left] {$+$};
\draw(-8,3) arc (-180:0:1);
\end{tikzpicture}
\een
Non-zero-weight spin configurations associated with the $N+1$ and $N$ sites are
\ben
\begin{tikzpicture}[scale=0.4]
  \draw[aline=0.65](-1,1) -- (3,1);
  \draw[aline=0.45](3,-1) -- (-1,-1);
\draw[aline=0.75](0,3) -- (0,-1);
\draw[aline=0.75](2,3) -- (2,-1);
\draw[aline=0.6](1,-2) -- (1,-3);
\draw(-0,3.5) node[above] {$+$};
\draw(2,3.5) node[above] {$+$};
\draw(3,1) node[right] {$-$};
\draw(3,-1) node[right] {$+$};
\draw(-1,1) node[left] {$-$};
\draw(-1,-1) node[left] {$+$};
\draw(1,-3) node[below] {$-$};
\draw(2,-2) node {$+$};
\draw(0,-2) node {$+$};
\draw(1, 1) node[above] {$-$};
\draw(1,-1) node[above] {$+$};
\draw(2, 0) node[right] {$+$};
\draw(0,0) node[left] {$+$};
\draw(1,-4) node[below] {$a^2b^2$};
\draw(2,-1) arc (-0:-180:1);
\end{tikzpicture}
\quad\quad
\begin{tikzpicture}[scale=0.4]
  \draw[aline=0.75](-1,1) -- (3,1);
  \draw[aline=0.4](3,-1) -- (-1,-1);
\draw[aline=0.4](1,2) -- (1,-3);
\draw(-3,0) node[left] {$-d^2$};
\draw(3,1) node[right] {$-$};
\draw(3,-1) node[right] {$+$};
\draw(2,3) arc (-0:-180:1);
\draw(1,-4) node[below] {$ab$};
\draw(-0,3.5) node[above] {$+$};
\draw(2,3.5) node[above] {$+$};
\draw(-1,1) node[left] {$-$};
\draw(-1,-1) node[left] {$+$};
\draw(1,-3) node[below] {$-$};
\draw(1, 0) node[right] {$-$};
\draw(1, 1.5) node[right] {$-$};
\end{tikzpicture},
\een
where again the weights are indicated beneath the pictures. 
Hence,  we have 
\ben q_{N+1,N} \, \cB^{(N+1)}-d^2 \cB^{(N)} \, q_{N+1,N}= ab(ab-d^2)E^{++}_{\;-} \ot B^{(N-1)},\een
and substitution into \eqref{eq:rrBinter} gives the desired relation \eqref{eq:rrB}. 

 Let now consider the proof of the more slightly complicated case of relations \eqref{eq:srrA} and  \eqref{eq:rrA} - the proofs for 
the corresponding $\cC$ and $\cD$ relations \eqref{eq:srrC}, \eqref{eq:rrC} and \eqref{eq:srrD},\eqref{eq:rrD} will then involve completely analogous arguments. 
We identify 
\ben
\begin{tikzpicture}[scale=0.4]
  \draw[aline=0.4](-9,1) -- (1,1);
  \draw[aline=0.62](1,-1) -- (-9,-1);
\foreach\x in {-8,-6,-4,-2,0}{
    \draw[aline=0.9](\x,2) -- (\x,-2);
  }
 \draw (1,1) arc (90:-90:1);
\draw(2,0)  node[bblob]{};
\draw(-12,0) node[left] {$\cA^{(N+1)}=$};
\draw(-8,2.5) node[above] {$N+1$};
\draw(-6,2.5) node[above] {$N$};
\draw(-4,2.5) node[above] {$\cdots$};
\draw(-2,2.5) node[above] {$2$};
\draw(-0,2.5) node[above] {$1$};
\draw(-9,1) node[left] {$+$};
\draw(-9,-1) node[left] {$+$};
\end{tikzpicture}
\een
and non-zero-weight configurations of spins around the leftmost column of R-matrices, with the two fixed spins shown, are
\ben\begin{tikzpicture}[scale=0.4]
  \draw[aline=0.8](-1,1) -- (1,1);
  \draw[aline=0.3](1,-1) -- (-1,-1);
\draw[aline=0.6](0,2) -- (0,-2);
\draw(1,1) node[right] {$+$};
\draw(1,-1) node[right] {$+$};
\draw(-1,1) node[left] {$+$};
\draw(-1,-1) node[left] {$+$};
\draw(0,2) node[above] {$+$};
\draw(0,-2) node[below] {$+$};
\draw(0,0) node[right] {$+$};
\draw(0,-3) node[below] {$a^2$};
 \end{tikzpicture}
\qquad 
\begin{tikzpicture}[scale=0.4]
  \draw[aline=0.8](-1,1) -- (1,1);
  \draw[aline=0.3](1,-1) -- (-1,-1);
\draw[aline=0.6](0,2) -- (0,-2);
\draw(1,1) node[right] {$+$};
\draw(1,-1) node[right] {$+$};
\draw(-1,1) node[left] {$+$};
\draw(-1,-1) node[left] {$+$};
\draw(0,2) node[above] {$-$};
\draw(0,-2) node[below] {$-$};
\draw(0,0) node[right] {$-$};
\draw(0,-3) node[below] {$b^2$};
 \end{tikzpicture}
\qquad 
\begin{tikzpicture}[scale=0.4]
  \draw[aline=0.8](-1,1) -- (1,1);
  \draw[aline=0.3](1,-1) -- (-1,-1);
\draw[aline=0.6](0,2) -- (0,-2);
\draw(1,1) node[right] {$-$};
\draw(1,-1) node[right] {$+$};
\draw(-1,1) node[left] {$+$};
\draw(-1,-1) node[left] {$+$};
\draw(0,2) node[above] {$-$};
\draw(0,-2) node[below] {$+$};
\draw(0,0) node[right] {$+$};
\draw(0,-3) node[below] {$ac$};
 \end{tikzpicture}
\qquad 
\begin{tikzpicture}[scale=0.4]
  \draw[aline=0.8](-1,1) -- (1,1);
  \draw[aline=0.3](1,-1) -- (-1,-1);
\draw[aline=0.6](0,2) -- (0,-2);
\draw(1,1) node[right] {$+$};
\draw(1,-1) node[right] {$-$};
\draw(-1,1) node[left] {$+$};
\draw(-1,-1) node[left] {$+$};
\draw(0,2) node[above] {$+$};
\draw(0,-2) node[below] {$-$};
\draw(0,0) node[right] {$+$};
\draw(0,-3) node[below] {$ac$};
 \end{tikzpicture}
\qquad 
\begin{tikzpicture}[scale=0.4]
  \draw[aline=0.8](-1,1) -- (1,1);
  \draw[aline=0.3](1,-1) -- (-1,-1);
\draw[aline=0.6](0,2) -- (0,-2);
\draw(1,1) node[right] {$-$};
\draw(1,-1) node[right] {$-$};
\draw(-1,1) node[left] {$+$};
\draw(-1,-1) node[left] {$+$};
\draw(0,2) node[above] {$-$};
\draw(0,-2) node[below] {$-$};
\draw(0,0) node[right] {$+$};
\draw(0,-3) node[below] {$c^2$};
 \end{tikzpicture}
\een
Hence we obtain Equation \eqref{eq:srrA} and \be
[\cA]^{(N+1)}&=&  
a^2  E^+_+  \ot [\cA]^{(N)} +  b^2  E^-_-  \ot [\cA]^{(N)}
+  ac  E^-_+ \ot [\cB]^{(N)} 
+ ac  E^+_- \ot [\cC]^{(N)}
+ c^2 E^-_- \ot  [\cD]^{(N)} \nn \\
&&+ (-1)^{N+1}\left\{q_{N+1,N} \, \cA^{(N+1)}-d^2 \cA^{(N)} \, q_{N+1,N}\right\}. 
 \label{eq:rrAinter}\ee
Continuing  as for the $\cB^{(N)}$ case, we have 
\ben &&q_{N+1,N} \, \cA^{(N+1)}-d^2 \cA^{(N)} \, q_{N+1,N}=\\[2mm]&&
\begin{tikzpicture}[scale=0.4]
  \draw[aline=0.4](-9,1) -- (1,1);
  \draw[aline=0.62](1,-1) -- (-9,-1);
\foreach\x in {-4,-2,0}{
    \draw[aline=0.9](\x,3) -- (\x,-3);
  }
\draw[aline=0.9](-6,3) -- (-6,-1);
\draw[aline=0.9](-8,3) -- (-8,-1);
\draw[aline=0.6](-7,-2) -- (-7,-3);
 \draw (1,1) arc (90:-90:1);
\draw(2,0)  node[bblob]{};
\draw(-8,3.5) node[above] {$N+1$};
\draw(-6,3.5) node[above] {$N$};
\draw(-4,3.5) node[above] {$\cdots$};
\draw(-2,3.5) node[above] {$2$};
\draw(-0,3.5) node[above] {$1$};
\draw(-9,1) node[left] {$+$};
\draw(-9,-1) node[left] {$+$};
\draw(-8,-1) arc (-180:0:1);
\end{tikzpicture}\quad\quad
\begin{tikzpicture}[scale=0.4]
\draw(-11,0) node[left] {$-d^2$};
  \draw[aline=0.4](-9,1) -- (1,1);
  \draw[aline=0.62](1,-1) -- (-9,-1);
\foreach\x in {-4,-2,0}{
    \draw[aline=0.9](\x,3) -- (\x,-3);
  }
\draw(-7,2) -- (-7,-3);
\draw[aline=0.6](-7,-2) -- (-7,-3);
 \draw (1,1) arc (90:-90:1);
\draw(2,0)  node[bblob]{};
\draw(-8,3.5) node[above] {$N+1$};
\draw(-6,3.5) node[above] {$N$};
\draw(-4,3.5) node[above] {$\cdots$};
\draw(-2,3.5) node[above] {$2$};
\draw(-0,3.5) node[above] {$1$};
\draw(-9,1) node[left] {$+$};
\draw(-9,-1) node[left] {$+$};
\draw(-8,3) arc (-180:0:1);
\end{tikzpicture}
\een
There are now considerably more non-zero-weight spin configurations associated with the $N+1$ and $N$ sites. They are
\ben
&&\begin{tikzpicture}[scale=0.4]
  \draw[aline=0.65](-1,1) -- (3,1);
  \draw[aline=0.45](3,-1) -- (-1,-1);
\draw[aline=0.75](0,3) -- (0,-1);
\draw[aline=0.75](2,3) -- (2,-1);
\draw[aline=0.6](1,-2) -- (1,-3);
\draw(-0,3.5) node[above] {$+$};
\draw(2,3.5) node[above] {$+$};
\draw(3,1) node[right] {$+$};
\draw(3,-1) node[right] {$+$};
\draw(-1,1) node[left] {$+$};
\draw(-1,-1) node[left] {$+$};
\draw(1,-3) node[below] {$-$};
\draw(2,-2) node {$+$};
\draw(0,-2) node {$+$};
\draw(1, 1) node[above] {$+$};
\draw(1,-1) node[above] {$+$};
\draw(2, 0) node[right] {$+$};
\draw(0,0) node[left] {$+$};
\draw(1,-4) node[below] {$a^4$};
\draw(2,-1) arc (-0:-180:1);
\end{tikzpicture}
\qquad\qquad
\begin{tikzpicture}[scale=0.4]
  \draw[aline=0.65](-1,1) -- (3,1);
  \draw[aline=0.45](3,-1) -- (-1,-1);
\draw[aline=0.75](0,3) -- (0,-1);
\draw[aline=0.75](2,3) -- (2,-1);
\draw[aline=0.6](1,-2) -- (1,-3);
\draw(-0,3.5) node[above] {$+$};
\draw(2,3.5) node[above] {$-$};
\draw(3,1) node[right] {$-$};
\draw(3,-1) node[right] {$+$};
\draw(-1,1) node[left] {$+$};
\draw(-1,-1) node[left] {$+$};
\draw(1,-3) node[below] {$-$};
\draw(2,-2) node {$+$};
\draw(0,-2) node {$+$};
\draw(1, 1) node[above] {$+$};
\draw(1,-1) node[above] {$+$};
\draw(2, 0) node[right] {$+$};
\draw(0,0) node[left] {$+$};
\draw(1,-4) node[below] {$a^3c$};
\draw(2,-1) arc (-0:-180:1);
\end{tikzpicture}
\qquad\qquad
\begin{tikzpicture}[scale=0.4]
  \draw[aline=0.65](-1,1) -- (3,1);
  \draw[aline=0.45](3,-1) -- (-1,-1);
\draw[aline=0.75](0,3) -- (0,-1);
\draw[aline=0.75](2,3) -- (2,-1);
\draw[aline=0.6](1,-2) -- (1,-3);
\draw(-0,3.5) node[above] {$-$};
\draw(2,3.5) node[above] {$+$};
\draw(3,1) node[right] {$-$};
\draw(3,-1) node[right] {$+$};
\draw(-1,1) node[left] {$+$};
\draw(-1,-1) node[left] {$+$};
\draw(1,-3) node[below] {$-$};
\draw(2,-2) node {$+$};
\draw(0,-2) node {$+$};
\draw(1, 1) node[above] {$-$};
\draw(1,-1) node[above] {$+$};
\draw(2, 0) node[right] {$+$};
\draw(0,0) node[left] {$+$};
\draw(1,-4) node[below] {$a^2bc$};
\draw(2,-1) arc (-0:-180:1);
\end{tikzpicture}\\[2mm]
&&\hspace*{-10mm}
\begin{tikzpicture}[scale=0.4]
  \draw[aline=0.75](-1,1) -- (3,1);
  \draw[aline=0.4](3,-1) -- (-1,-1);
\draw[aline=0.4](1,2) -- (1,-3);
\draw(-3,0) node[left] {$-d^2$};
\draw(3,1) node[right] {$+$};
\draw(3,-1) node[right] {$+$};
\draw(2,3) arc (-0:-180:1);
\draw(1,-4) node[below] {$b^2$};
\draw(-0,3.5) node[above] {$+$};
\draw(2,3.5) node[above] {$+$};
\draw(-1,1) node[left] {$+$};
\draw(-1,-1) node[left] {$+$};
\draw(1,-3) node[below] {$-$};
\draw(1, 0) node[right] {$-$};
\draw(1, 1.5) node[right] {$-$};
\end{tikzpicture}\quad
\begin{tikzpicture}[scale=0.4]
  \draw[aline=0.75](-1,1) -- (3,1);
  \draw[aline=0.4](3,-1) -- (-1,-1);
\draw[aline=0.4](1,2) -- (1,-3);
\draw(-3,0) node[left] {$-d^2$};
\draw(3,1) node[right] {$-$};
\draw(3,-1) node[right] {$+$};
\draw(2,3) arc (-0:-180:1);
\draw(1,-4) node[below] {$ac$};
\draw(-0,3.5) node[above] {$+$};
\draw(2,3.5) node[above] {$+$};
\draw(-1,1) node[left] {$+$};
\draw(-1,-1) node[left] {$+$};
\draw(1,-3) node[below] {$+$};
\draw(1, 0) node[right] {$+$};
\draw(1, 1.5) node[right] {$-$};
\end{tikzpicture}\quad
\begin{tikzpicture}[scale=0.4]
  \draw[aline=0.75](-1,1) -- (3,1);
  \draw[aline=0.4](3,-1) -- (-1,-1);
\draw[aline=0.4](1,2) -- (1,-3);
\draw(-3,0) node[left] {$-d^2$};
\draw(3,1) node[right] {$-$};
\draw(3,-1) node[right] {$-$};
\draw(2,3) arc (-0:-180:1);
\draw(1,-4) node[below] {$c^2$};
\draw(-0,3.5) node[above] {$+$};
\draw(2,3.5) node[above] {$+$};
\draw(-1,1) node[left] {$+$};
\draw(-1,-1) node[left] {$+$};
\draw(1,-3) node[below] {$-$};
\draw(1, 0) node[right] {$+$};
\draw(1, 1.5) node[right] {$-$};
\end{tikzpicture}
\een
Hence we have 
\ben q_{N+1,N} \, \cA^{(N+1)}-d^2 \cA^{(N)} \, q_{N+1,N}&=& a^4\, E^{++}_{\;-} \ot A^{(N-1)} + a^3c\, 
 E^{+-}_{\;-} \ot B^{(N-1)}  + a^2bc\,  E^{-+}_{\;-} \ot B^{(N-1)}\\
&&-b^2 d^2 \, E^{++}_{\;-} \ot A^{(N-1)} - ac d^2 \, E^{++}_{\;+} \ot B^{(N-1)} - c^2 d^2 \, E^{++}_{\;-} \ot D^{(N-1)},\een
and substitution into \eqref{eq:rrAinter} gives the desired relation \eqref{eq:rrA}. 

Having established relations \eqref{eq:srrA}--\eqref{eq:rrD}, we will now use them to prove the statement of the theorem.
Let us make the inductive hypothesis that \eqref{eq:ft1}--\eqref{eq:ft4} hold for a specific integer $N\geq 2$, i.e., we have
\be
\, [ \cA]^{(N)}&=& (-1)^N c\, d\,  E^+\ot  \cB^{(N-1)}\label{eq:mft1},\\
\, [ \cB]^{(N)}&=& 0\label{eq:mft2},\\
\, [ \cC]^{(N)}&=&(-1)^N \left( a\, c\,  E^+\ot \cA^{(N-1)} +c^2\,  E^- \ot\cB^{(N-1)} +  c\, d\,  E^+ \ot\cD^{(N-1)} \right)\label{eq:mft3},\\
\, [ \cD]^{(N)}&=& (-1)^N  b\, c\,  E^+\ot \cB^{(N-1)}\label{eq:mft4}.
\ee 
We consider first the inductive step for $[\cB]^{(N)}$: from \eqref{eq:rrB}, and the above inductive hypothesis for $[ \cA]^{(N)}$,
$[ \cB]^{(N)}$ and $[ \cD]^{(N)}$ we have 
 \ben [\cB]^{(N+1)}&=&bc  E^+_- \ot [\cA]^{(N)}
+ bc E^+_- \ot [\cD]^{(N)}
+ (-1)^N ab(d^2-ab) \, E^{++}_{\;-} \ot \cB^{(N-1)},\\
&=& (-1)^N \left\{ bc^2d+ b^2c^2 + ab(d^2-ab)\right\} E^{++}_{\;-} \ot \cB^{(N-1)} =0, \een
with the expression vanishing due to the two identities \eqref{eq:ids}.

Let us again go through the inductive step for the more complicated $\cA^{(N)}$ case and leave the completely analogous proofs for $\cC^{(N)}$ and $\cD^{(N)}$ cases to the reader.
Using \eqref{eq:rrA} and the inductive hypothesis \eqref{eq:mft1}--\eqref{eq:mft4} gives
\be [\cA]^{(N+1)} &=& (-1)^N \left\{ a^2 c d \, E^{++}_{\;+} \ot \cB^{(N-1)} +b^2 cd \, E^{-+}_{\;-} \ot \cB^{(N-1)} \right.\nn 
\\ &&\hspace*{-8mm}+ ac \left( a\, c\, E^{++}_{\;-} \ot \cA^{(N-1)} +c^2\, E^{+-}_{\;-} \ot \cB^{(N-1)} +  c\, d\,  
 E^{++}_{\;-} \ot \cD^{(N-1)}\right)
+bc^3  E^{-+}_{\;-} \ot \cB^{(N-1)}
\nn
\\&&\hspace*{-8mm}+(b^2\,d^2\,-a^4) \, E^{++}_{\;-} \ot \cA^{(N-1)} + acd^2 \,E^{++}_{\;+} \ot \cB^{(N-1)} \nn \\ &&\hspace*{-8mm} \left.
 - a^2 b c \, E^{-+}_{\;-} \ot \cB^{(N-1)}
- a^3 c  E^{+-}_{\;-} \ot \cB^{(N-1)} + c^2 d^2 E^{++}_{\;-} \ot \cD^{(N-1)}\right\}.
 \ee
Using the identities \eqref{eq:ids}, this expression simplifies to 
 \be \hspace*{-3mm}[\cA]^{(N+1)} =\hspace*{-1mm} (-1)^{N+1} cd \left( bc\, E^{++}_{\;-} \ot \cA^{(N-1)}  +
 ab \,E^+\ot \id\ot \cB^{(N-1)}  + bc\,  E^{++}_{\;-} \ot D^{(N-1)}\right).\ee
Finally, using equation \eqref{eq:srrB} gives
\be [\cA]^{(N+1)} &=& (-1)^{N+1} cd\,E^+ \ot \cB^{(N)},\ee
which completes the inductive step for \eqref{eq:mft1}.
\end{proof}

It is a simple corollary of Theorem 1 that  the transfer matrix 
\ben t^{(N)}(u)=  \frac{\sinh(u)}{\sinh(\eta)} \cA^{(N)}(u) - \frac{\sinh(u+2\eta)}{\sinh(\eta)}\cD^{(N)}(u),\een
commutes with the SUSY operator $Q^{(N)}$ thus
\ben Q^{(N)} t^{(N)}(u)= d(u)^2 \,t^{(N-1)}(u) \,Q^{(N)}.\een
Differentiating, and recalling that $\eta=2\pi i/3$, we find
\ben  Q^{(N)} t^{(N)'}(0)- t^{(N-1)'}(0)Q^{(N)} = \frac{-2i}{\sqrt{3}} t^{(N-1)}(0) Q^{(N)} =\frac{2i}{\sqrt{3}}  Q^{(N)} .\een
Using Equation \eqref{eq:tHcon}, this implies that 
\be H^{(N-1)} Q^{(N)}=Q^{(N)} H^{(N)},\ee
which is of course the SUSY relation derived directly in Section 2.

Using Theorem \ref{thm1} leads to the result
\ben
Q^{(N)} \cB^{(N)}(\lambda_1)\cB^{(N)}(\lambda_2) \cdots \cB^{(N)}(\lambda_m) \, \Omega^{(N)}
= \prod_{i=1}^m d(\gl_i)^2 \, \cB^{(N-1)}(\lambda_1)\cB^{(N-1)}(\lambda_2) \cdots \cB^{(N-1)}(\lambda_m) Q^{(N)} \, \Omega^{(N)}.\een
Next, we note that 
\be Q^{(N)} \, \Omega^{(N)} = \sli_{i=1}^{N-1} (-1)^{i+1}\sigma_i^{-}\Omega^{(N-1)} =(-1)^N  \cB^{(N-1)}(\eta) \Omega^{(N-1)}.\label{eq:QBreln}\ee
The left-hand equality follows immediately from the definition \eqref{eq:Qdef} of $Q^{(N)}$. The right-hand equality comes from noting that $a(\eta)=-1$, $b(\eta)=c(\eta)=1$ and 
\ben K^-(\eta)=\begin{pmatrix} 0&0\\0&1\end{pmatrix}.\een
 Hence the non-zero-weight configuration contributions to $\cB^{(N-1)}(\eta) \Omega^{(N-1)}$ are 
\ben
\begin{tikzpicture}[scale=0.4]
  \draw[aline=0.38](-9,1) -- (1,1);
  \draw[aline=0.66](1,-1) -- (-9,-1);
\foreach\x in {-8,-6,-4,-2,0}{
    \draw[aline=0.9](\x,2) -- (\x,-2);
  }
 \draw (1,1) arc (90:-90:1);
\draw(2,0)  node[bblob]{};
\draw(-12,0) node[left] {$\sli_{i=1}^{N-1}$};
\draw(-8,3) node[above] {$N-1$};
\draw(-6,3) node[above] {$\cdots$};
\draw(-4,3) node[above] {$i$};
\draw(-2,3) node[above] {$\cdots$};
\draw(-0,3) node[above] {$1$};
\draw(-9,1) node[left] {$-$};
\draw(-9,-1) node[left] {$+$};
\draw(-6.3,0.7) node[left] {$-$};
\draw(-6.3,-1.3) node[left] {$+$};
\draw(-4.3,0.7) node[left] {$-$};
\draw(-4.3,-1.3) node[left] {$+$};
\draw(-2.3,0.7) node[left] {$-$};
\draw(-2.3,-1.3) node[left] {$-$};
\draw(-0.3,0.7) node[left] {$-$};
\draw(-0.3,-1.3) node[left] {$-$};
\draw(1.7,0.7) node[left] {$-$};
\draw(1.7,-1.3) node[left] {$-$};
\draw(-7.8,0) node[left] {$+$};
\draw(-5.8,0) node[left] {$+$};
\draw(-3.8,0) node[left] {$+$};
\draw(-1.8,0) node[left] {$+$};
\draw(0.2,0) node[left] {$+$};
\draw(-8,1.9) node[above] {$+$};
\draw(-8,-1.9) node[below] {$+$};
\draw(-6,1.9) node[above] {$+$};
\draw(-6,-1.9) node[below] {$+$};
\draw(-4,1.9) node[above] {$+$};
\draw(-4,-1.9) node[below] {$-$};
\draw(-2,1.9) node[above] {$+$};
\draw(-2,-1.9) node[below] {$+$};
\draw(-0,1.9) node[above] {$+$};
\draw(-0,-1.9) node[below] {$+$};
\end{tikzpicture}
\een
and the weight of the $i$th contribution to this sum is $a^{N-1-i} b^{N-2+i} c= (-1)^{N-1-i}$.  This yields the right-hand equality in \eqref{eq:QBreln}.
Thus it follows that we have the key relation
\be 
&&Q^{(N)} \cB^{(N)}(\lambda_1)\cB^{(N)}(\lambda_2) \cdots \cB^{(N)}(\lambda_m) \, \Omega^{(N)}
\nn\\&&= (-1)^N \prod_{i=1}^m d(\gl_i)^2 \, \cB^{(N-1)}(\lambda_1)\cB^{(N-1)}(\lambda_2) \cdots \cB^{(N-1)}(\lambda_m) \cB^{(N-1)}(\eta)  \, \Omega^{(N-1)}.\label{eq:keyreln}\ee
That is, the SUSY operator $Q^{(N)}$ takes an off-shell Bethe state $\cB(\gl_1)\cdots \cB(\gl_m) \Omega^{(N)}$, with $\gl_i\neq \eta$ for $\in\{1,2,\cdots,m\}$,  of the system with length $N$ to an off-shell Bethe state of the system with size $N-1$.
If any of the $\gl_i=\eta$ then we have  $d(\gl_i)=0$ and the off-shell Bethe state $\cB(\gl_1)\cdots \cB(\gl_m) \Omega^{(N)}$ is in the kernel of $Q^{(N)}$. The nilpotency property  $Q^{(N-1)} Q^{(N)}$  on any off-shell Bethe state follows immediately from this.

Let us now consider the action of $Q^{(N)}$ on on-shell Bethe states. We have the following:
\begin{proposition}
If $\{\gl_1,\gl_2,\cdots,\gl_m\}$, with $\gl_i\neq \eta$ for $i\in\{1,2,\cdots,m\}$, satisfy the Bethe equations \eqref{eq:Beqns} for a size $N$ system, , then $\{\gl_1,\gl_2,\cdots,\gl_m,\gl_{m+1}:=\eta\}$ satisfy the Bethe equations for a size $N-1$ system.
\end{proposition}
\begin{proof}
First of all, it is a simple consequence of \eqref{eq:srrA} and \eqref{eq:srrD} that we have\ben
\cA^{(N)}(u) \Omega^{(N)} &=& E^+_+ \ot  a(u)^2\cA^{(N-1)}(u) \Omega^{(N-1)}, \\
\cD^{(N)}(u) \Omega^{(N)} &=& E^+_+ \ot\left[c(u)^2 \cA^{(N-1)}(u) + b(u)^2 \cD^{(N-1)}(u)\right] \Omega^{(N-1)}.
\een
from which it follows that $\Delta_\pm^{(N)}(u)$ defined by the relations \eqref{eq:Deltadef} satisfy 
\be \Delta^{(N)}_+(u) = a(u)^2 \Delta^{(N-1)}_+(u)\,\quad 
\Delta^{(N)}_-(u) = b(u)^2 \Delta^{(N-1)}_-(u).\label{eq:dpmrr}\ee 
These recursion relations can be solved to give the explicit expressions
\be \Delta^{(N)}_+(u) =  \frac{a(u)^{2N}\sinh(u-\eta)}{\sinh(\eta)},\quad \Delta^{(N)}_-(\gl) = 
- \frac{b(u)^{2N}\sinh(2u)\sinh(u+2\eta)}{\sinh(\eta)}.\label{eq:dpm}\ee
From relations \eqref{eq:dpmrr}, we see that if $\{\gl_1,\gl_2,\cdots,\gl_m\}$ satisfy the lattice-length $N$ Bethe equations \eqref{eq:Beqns}, then for $j\in \{1,2,\cdots,m\}$ we also have
\ben -\sinh(2\gl_j) \frac{\Delta_+^{(N-1)}(\gl_j)}{\Delta_-^{(N-1)}(\gl_j)}&=& 
\left[\prod_{k=1,\cdots, m;\,k\neq j} \frac{\sinh(\gl_j-\gl_k+\eta)\sinh(\gl_j+\gl_k+2\eta) }
{\sinh(\gl_j-\gl_k-\eta)\sinh(\gl_j+\gl_k)}\right] \frac{\sinh^2(\gl_j)}{\sinh^2(\gl_j+\eta)},\\
&=& \prod_{k=1\cdots m+1;\,k\neq j} \frac{\sinh(\gl_j-\gl_k+\eta)\sinh(\gl_j+\gl_k+2\eta) }
{\sinh(\gl_j-\gl_k-\eta)\sinh(\gl_j+\gl_k)},
\een
where $\gl_{m+1}:=\eta$.
For the remaining root $\gl_{m+1}$, we have
 \ben 
-\sinh(2\gl_{m+1}) \frac{\Delta_+^{(N-1)}(\gl_{m+1})}{\Delta_-^{(N-1)}(\gl_{m+1})}= 1 = 
\prod_{k=1\cdots m} \frac{\sinh(\gl_{m+1}-\gl_k+\eta)\sinh(\gl_{m+1}+\gl_k+2\eta) }
{\sinh(\gl_{m+1}-\gl_k-\eta)\sinh(\gl_{m+1}+\gl_k)}.\een 
\end{proof}

Thus we see that $Q^{(N)}$ takes an on-shell Bethe state for a size $N$ lattice with roots $\{\gl_1,\gl_2,\cdots,\gl_m\}$, with $\lambda_i\neq \eta$, to an on-shell Bethe state  $\{\gl_1,\gl_2,\cdots,\gl_m,\eta\}$ for a size $N-1$ lattice. It is simple to check from the explicit expression \eqref{eq:teigenval} that the transfer-matrix eigenvalues for the two different-sized chains are related by 
\ben \tau^{(N)}(u;\gl_1,\gl_2,\cdots,\gl_m) = d^2(u) \,\tau^{(N-1)}(u;\gl_1,\gl_2,\cdots,\gl_m,\eta).\een
It then follows from \eqref{eq:tHcon} that the corresponding eigenvalues of $H^{(N)}$ and $H^{(N-1)}$ are equal.

\section{Discussion}
In this paper we have computed the action of the lattice SUSY operator $Q^{(N)}$ on off-shell and on-shell Bethe states. The key result is given by Equation \eqref{eq:keyreln} from which Proposition 2 follows. Proposition 2 is in turn analogous to the observations about pairings of Bethe roots made for closed spin chains in \cite{MR2903041}. One significant difference is that the pairings for closed chains were observed only for Bethe states in particular momentum sectors of the Hilbert space in  \cite{MR2903041} whereas in the open chain considered here the pairings occur for all Bethe states. 

Two possible uses of our results are as follows: firstly,  by rewriting the action of $Q^{(N)}$ in terms of $\cB(\eta)$ (at least on Bethe states) we have connected the SUSY with the underlying reflection algebra - which is in turn a coideal subalgebra of the quantum affine algebra $U_q(\widehat{\mathfrak{sl}}_2)$ with $q=\exp(\eta)$ \cite{MR1193836}. We view this as a first step in the development of a more systematic understanding of the appearance and role of lattice SUSY within the context of the QISM approach. 

Secondly, a major outstanding problem is to compute the exact form of the zero-energy, vacuum vector $\omega^{(N)}$ of the Hamiltonian \eqref{eq:XXZ1}. Such a vector $\omega^{(N)} \in V^{\ot N}$ is characterised by the requirement  \be Q^{(N)}\omega^{(N)}=0, \quad Q^{(N+1)\dag} \omega^{(N)}=0.\label{eq:cohom}\ee  
So $\omega^{(N)}$ is a SUSY singlet, or equivalently $\omega^{(N)}$ is a vector in $\Ker(Q^{(N)})$ whose inner product with any vector in $\Ima(Q^{(N+1)})$ is zero.
An important step towards finding $\vac$ was taken recently in \cite{MR3635719}, in which it is shown firstly that the singlet is unique and in the $N$ mod 2 spin sector, and secondly that it is of the form
\ben \vac=\chi^{(N)}+\psi^{(N)},\een
where $\chi^{(N)}\in\hbox{Ker}(Q^{(N)})$ was constructed explicitly, but $\psi^{(N}\in \Ima(Q^{(N+1)})\subset V^{\ot N}$ was an unknown vector in the image space of $Q^{(N+1)}$. Given this result, a sufficent requirement for finding $\psi^{(N}$ (and hence $\omega^{(N)}$) is to have an explicit basis $
\{b_i^{(N)}\}$ for $\Ima(Q^{(N+1)})$ - since in this case we have 
\ben (\omega^{(N)},b_i^{(N)}) =0 = (\chi^{(N)},b_i^{(N)})+ (\psi^{(N)},b_i^{(N)}), \een
which determines $\psi^{(N}$. A first step to finding such a basis is to find the dimensions \ben \kappa^{(N)}:=\hbox{dim} \left(\hbox{Ker}(Q^{(N)})\right)\quad 
\iota^{(N)}:=\hbox{dim} \left(\hbox{Im}(Q^{(N)})\right).\een
This is a simple exercise as we have the two relations $\kappa^{(N)}=\iota^{(N+1)}+1$ (from the result concerning the uniqueness of the singlet given in \cite{MR3635719}) and $\kappa^{(N)}+\iota^{(N)}=
\hbox{dim}(V^{\ot N})=2^N$ (from elementary linear algebra). Solving these relations yields
\ben \iota^{(N)}= \lceil \frac{2}{3} (2^{N-1}-1)\rceil.\een
 This means that the $2^N$ dimensional eigenspace of $H^{(N)}$ is made up as follows:
\begin{description}
\item  1 zero energy SUSY singlet - the  vacuum state, \item $ \lceil \frac{2}{3} (2^{N}-1)\rceil$  states paired with a SUSY partner which is an eigenstate of $H^{(N+1)}$,
\item 
$\lceil \frac{2}{3} (2^{N-1}-1)\rceil$ states paired with a SUSY partner which is an eigenstate of $H^{(N-1)}$.\end{description}
This result is consistent with numerical diagonalisation of $H^{(N)}$ for small lattices. While we have yet to find an explicit expression for the basis vectors $ \{b_i^{(N)}\}$ of $\Ima(Q^{(N+1)})$ for general $N$, the suggestive combinatorial nature of their number $\iota^{(N+1)}$, and also of the number of such states in the spin $N$ mod 2 subspace of the vacuum vector $\omega^{(N)}$,  makes up optimistic that this is indeed possible.

Our results also provide a different way of describing $\omega^{(N)}$: it should be an on-shell Bethe vector 
$\cB(\gl_1) \cB(\gl_2) \cdots \cB(\gl_m) \Omega^{(N)}\in \Ker(Q^{(N)})$ for which $\gl_i\neq \eta$ for all $i\in\{1,2,\cdots,m\}$. We hope that this characterisation might also be useful  in  finding an exact expression for $\vac$.

\subsection*{Acknowledgements}
We would like to thank Anastasia Doikou, Christian Hagendorf, Des Johnston and Paul Zinn-Justin for useful discussions.


\end{document}